\documentclass[submission,copyright,creativecommons]{eptcs}
\providecommand{\event}{Non-Classical Logics 2024 (NCL'24)}

\ifpdf
  \usepackage{underscore}         
  \usepackage[T1]{fontenc}        
\else
  \usepackage{breakurl}           
\fi

\title{Two Cases of Deduction with Non-referring Descriptions}
\author{Ji\v{r}\'{i} Raclavsk\'{y} 
\institute{Masaryk University\\ Brno, Czech Republic}
\email{raclavsky@phil.muni.cz}
}
\def\titlerunning{Non-referring Descriptions and Deduction}
\def\authorrunning{J. Raclavsk\'{y}}

\usepackage{amsmath, amscd, amsthm, amssymb, mathrsfs,amsfonts}
\usepackage{bm}
\usepackage{bussproofs}
\usepackage{xcolor}
\usepackage{graphicx}

\newcommand{\rotatediota}{{\mathpalette\rotiota\relax}}
\newcommand{\rotiota}[2]{\rotatebox[origin=c]{180}{$#1\iota$}}

\DeclareMathAlphabet{\mathsfit}{T1}{\sfdefault}{\mddefault}{\sldefault}
\SetMathAlphabet{\mathsfit}{bold}{T1}{\sfdefault}{\bfdefault}{\sldefault}

\newcommand{\acq}[1]{\ulcorner{#1}\urcorner}
\newcommand{\acqb}[1]{\textsf{\textbf{#1}}}
\newcommand{\acqs}[1]{\bm{\pmb{#1}}}

\newcommand{\cc}{\mathsfit}
\newcommand{\obj}[1]{\mathtt{#1}}

\usepackage{amsthm}
\theoremstyle{definition}
\newtheorem{theorem}{Theorem}
\newtheorem{definition}{Definition}
\newtheorem{lemma}{Lemma}

\begin{document}
\maketitle

\begin{abstract}
Formal reasoning with non-denoting terms, esp. non-referring descriptions such as ``the King of France'', is still an under-investigated area. The recent exception being a series of papers e.g. by Indrzejczak and Zawidzki. The present paper offers an alternative to their approach since instead of free logic and sequent calculus, it's framed in partial type theory with natural deduction in sequent style. Using a Montague- and Tich\'{y}-style formalization of  natural language, the paper successfully handles deduction with intensional transitives whose complements are non-referring descriptions, and derives Strawsonian rules for existential presuppositions of sentences with such descriptions.
\end{abstract}


\section{Introduction}

In his groundbreaking 1905 paper ``\textrm{On Denoting}'', Russell \cite{russell1905} offered a widely adopted theory of \textit{(definite) descriptions}, i.e. the \textit{singular terms} of the form ``\textsl{the $F$}'', the most famous example being ``\textsl{the King of France}''. 
Russell rightly indicated that

\begin{enumerate}
\small
\item[1.]
Each (definite) description is satisfied by at most one entity. 
\hfill (\textit{Uniqueness})
\item[2.]
Descriptions typically involve predicative (some say: descriptive) content. 
\hfill (\textit{Predicativity})
\end{enumerate}

\noindent
Which has been generally accepted, cf. e.g. Ludlow \cite{ludlow2023}. 
But the true brilliance of Russell's theory lies in its capability to 
handle even the fact that

\begin{enumerate}
\small
\item[3.]
Some descriptions (e.g. ``\textsl{the King of France}'')
are \textit{non-referring}.
\hfill
(\textit{Non-Referring Descriptions})
\end{enumerate}

However,
Russell's own elaboration of \textit{formal semantics} of descriptions became divisive. On one side, many theoreticians praised Russell for paradigmatic philosophical analysis -- which states that

\begin{enumerate}
\small
\item[(r1)]
Descriptions have no meaning in isolation, so ``\textsl{the King of France}'' is {meaningless} \textit{per se}.
\item[(r2)]
Descriptions only contribute to sentence's meaning by scattered bits such as the meaning of ``$F$''.
\item[(r3)]
The sentential meaning of e.g. ``\textsl{The King of France is bald}'' is to be reconstructed in terms of first-order logic with identity as an \textit{existential statement} of the form $\exists x ( F(x) \land G(x) \land \forall y ( F(y) \to y=x) )$. 
\end{enumerate}

\noindent
While (r2) has rarely been challenged since it obviously matches Point 2, 
(r3)'s consequence that 
sentences with descriptions in `referential positions' 
 (cf. e.g. ``\textsl{The King of France is an $F$.}'')
are definitely true or false (which was seen as an advantage by Russell and some his allies) has been persistently criticized by Strawson \cite{strawson1950} and his numerous supporters. 

But the clash between Strawson and Russell as regards (r3) overshadows the fact that both Russell and Strawson were followed by many writers (e.g. Tich\'{y} \cite{tichy1986}, Farmer \cite{farmer1990}, Feferman \cite{feferman1995}, Indrzejczak and Zawidzki \cite{indrzejczak-zawidzki2023}) who \textit{did} adopt Point 3 (neglecting here Strawson's stress on \textit{use} of descriptions).
The corresponding area of research is now known as the \textit{logic of non-denoting terms}, or, more generally, 
as \textit{partial logic}.
For an introduction, see e.g. Farmer \cite{farmer1990},  Feferman \cite{feferman1995}, or the present author's \cite{kuchynka-raclavsky2024}. 

\textit{Non-denoting terms} are in fact ubiquitous in 

\begin{enumerate}
\small
\item[a.]
(formalized) mathematics, cf. e.g. ``$3\div 0 $'', ``$\sqrt{x}$'' (for negative $x$), ``$\underset{x\mapsto a}{\mathrm{lim}}$ $f(x)$'' (for some values);
\item[b.]
natural language, cf. e.g. ``\textsl{the greatest prime}'', ``\textsl{the King of France}'';
\item[c.]
computer science, cf. e.g. abortive halting programs, unsuccessful database searches, etc.
\end{enumerate}

\noindent
Yet in (philosophical) logic such \textit{partiality phenomena} have been largely abandoned. In particular, many logical textbooks and related writings offer no sufficient discussion of descriptions and simply reiterate  Russell's controversial points (r1) and (r3). 
But once we overview further literature, we find various broadly Strawsonian approaches; they roughly fit the following quadruple of views:

\begin{enumerate}
\setlength\itemsep{.5em}
\small
\item[(s1)]
Descriptions $D$ do have a self-sustaining meaning: either (s1.a) $D$'s meaning is identical with $D$'s reference/denotation, or 
(s1.b) $D$'s meaning determines $D$'s reference/denotation.
\item[(s2)]
Sentences with descriptions $D$ in `referential position' are either (s2.a) implicitly existential claims, or (s2.b) are in no sense existential claims.
\end{enumerate}

\textit{Free logic} ($\mathsf{FL}$) seems to provide the largest platform for positions revolving mainly on (s2)-topics, cf. e.g. Bencivenga \cite{bencivenga2002}. As repeatedly argued by its proponents, 
$\mathsf{FL}$ delivers desired truth conditions for sentences with descriptions and other singular terms in `referential position' regardless their actual reference. Some writers, e.g. Farmer \cite{farmer1990}, Fitting and Mendelssohn \cite{fitting-mendelsohn1998}, follow Frege \cite{frege1892} and Scott \cite{scott1979} and maintain that $D$ refers to a \textit{dummy value} (sometimes denoted $\bot_\tau$ or $*_\tau$), an artificially chosen object either from the `domain we live in' (sometimes identified with \textit{inner domain}), or some \textit{outer domain}. 
Some writers at least briefly discuss so induced \textit{existential commitments} (i.e. s2.a), but many (e.g. Blamey \cite{blamey1986}) consider dummy values being mere technical devices. 
On the other hand, some theoreticians, e.g. Lehmann \cite{lehmann2002}, Tich\'{y} \cite{tichy1988} and also the present writer, rather favour the view that

\begin{enumerate}
\small
\item[4.]
Non-referring descriptions refer to nothing whatsoever (i.e. not to dummy entities). 
\hfill
(\textit{Genuine Partiality})
\end{enumerate}

\noindent
Whereas Occam's Principle of Parsimony provides a potent argument in favour of such a position.

Another assumption of the present paper, which is now widely adopted in literature, is an overt dismissal of Russell's (r1):

\begin{enumerate}
\small
\item[5.]
Descriptions have meaning even in isolation.
\hfill 
(\textit{Descriptions' Meaning})
\end{enumerate}

\noindent
As argued on numerous places in literature, in particular by 
Tich\'{y} \cite{tichy1971,tichy1988}, Montague \cite{montague1973}, Fitting and Mendelsohn \cite{fitting-mendelsohn1998}, Indrzejczak and Zawidzki \cite{indrzejczak2020,indrzejczak-zawidzki2023}, Orlandelli \cite{orlandelli2021}, and even the present author \cite{raclavsky2020},

\begin{enumerate}
\small
\item[6.]
The reference of `empirical' descriptions such as ``\textsl{the King of France}'' is a contingent affair, i.e. the reference of expressions depends on possible worlds and time instants. 
\qquad\qquad\qquad\;
(\textit{Modality, Temporality})
\end{enumerate}

\noindent
Moreover, the present paper relies on arguments developed by 
Tich\'{y} (e.g. \cite{tichy1988}, Moschovakis \cite{moschovakis2006} and others  (incl. the present author's \cite{raclavsky2020,kuchynka-raclavsky2024}) 
in favour of the view that

\begin{enumerate}
\small
\item[7.]
Meanings of descriptions are algorithmic computations that determine possible-worlds intensions.
\\
\medskip
\hfill 
(\textit{Algorithmic Meanings})
\end{enumerate}

\noindent
Note that Point 7 sustains the \textit{Principle of Compositionality}: the meaning of a compound expression $E$ depends on the meaning of $E$'s parts -- regardless their contingent reference (if any).

\subsection{Problems addressed in the present paper}

So far we have sketched an overall background of our investigation; now it's time for a brief and informal discussion of problems addressed in this paper, indicating also their solution elaborated below.

\textit{Problem 1}. In his \cite{church1951}, 
Church published a decisive counter-argument against Russell's theory of descriptions. It employs so-called \textit{intensional transitive verbs} (ITVs) such as ``\textsl{seek}'', cf. e.g. the sentence

\begin{itemize}
\item[]
``\textsl{Ponce de León searched for the Fountain of Youth}''.
\end{itemize}

\noindent
As correctly observed by Church, and emphasised by Quine in his seminal paper \cite{quine1956}, such sentences \textit{lack existential commitment} as regards complements of ITVs. The sought object need not to exist, so we are not allowed to derive that (say) the Fountain of Youth exists.
Yet such a fallacious inference is not prevented by Russell's theory (since no discrimination between primary/secondary occurrence of a description can be employed here as in case of propositional attitudes). Which thus presents its fatal flaw. 

Church \cite{church1951} noted that Frege's theory of singular terms is therefore superior to Russell's, since it can reject undesired inferences by pointing out the confusion of reference (\textit{Bedeutung}) and sense (\textit{Sinn}). The distinction was elaborated by Carnap \cite{carnap1947} and other adherents of \textit{possible-worlds semantics} ($\textsl{PWS}$) in terms of extensions and possible-worlds \textit{intensions} (i.e. certain functions to extensions). Intensions such as the \textit{individual concept} of the Fountain of Youth figure as complement objects of the relations(-in-intensions) which are meanings of ITVs, cf. Tich\'y  \cite{tichy1971,tichy1988}, Montague \cite{montague1973}, or e.g. \cite{raclavsky2020}. 

The widely adopted 
solution, and even the problem itself, is surprisingly entirely missing in recent studies on reasoning with descriptions (cf. e.g. \cite{fitting-mendelsohn1998,indrzejczak-zawidzki2023}). One of the aims of the present paper is to suggest (on a particular example of a chosen deduction system) that any logical framework adopting $\textsl{PWS}$ can successfully cope with Problem 1. Of course, a $\mathsf{FL}$ restricted to first-order quantification is not useful here, since adoption of PWS-intensions typically amounts to adoption of \textit{quantification over functions} and so  \textit{higher-order logic} $\mathsf{HOL}$ -- e.g. the \textit{type theory} $\mathsf{TT^*}$ \cite{raclavsky2020,raclavsky2022,kuchynka-raclavsky2024} utilised below.

\textit{Problem 2.}
For investigation of Problem 1, the logical system $\mathsf{TT^*}$ deployed below might be perhaps seen as over-dimensioned. But it's deduction system 
-- a natural deduction in sequent style $\mathsf{ND}_\mathsf{TT^*}$, \cite{raclavsky2020,raclavsky2022,kuchynka-raclavsky2024,tichy1982} -- 
 is a \textit{labelled calculus}, for which esp. Gabbay \cite{gabbay1996} provided an extensive argumentation. 
In particular, a part of the present paper shows how labelled (or `signed') formulas allow to control inference even in cases the formulas being non-denoting expressions. (To avoid misunderstanding: according to the present approach all well-formed expressions always have certain meaning, viz. an algorithmic computation, yet they may lack a reference/denotation.)

Being so equipped, a formal reconstruction of Strawson's \cite{strawson1950} `logic' of \textit{existential presupposition} is possible. We will, for example, derive  an exact logical rule of $\mathsf{ND_{TT^*}}$ that corresponds to Strawson's claim (p. 330) that

\begin{itemize}
\item[]
If the sentence ``\textsl{The King of France doesn't exist}'' is false, then the sentence ``\textsl{The King of France is (not) bald}'' is without a truth value.
\end{itemize}

\noindent
Albeit such Strawsonian reasoning is considered sound by many linguists and some philosophers of language, its formal reconstruction seems to be entirely missing in logical literature.

\textit{Structure of the paper}.
In Secs. 2 and 3, we expose the \textit{partial type theory} $\mathsf{TT^*}$ and a natural deduction system for it, $\mathsf{ND_{TT^*}}$.
In Sec. 4, we first show how to formalize meanings of descriptions and expressions involving them, and how to formally check natural language arguments. We test the proposal against two groups of frequently neglected inferences, namely (a) inferences with intensional transitives (well understood in formal semantics), (b) Strawsonian inferences (rarely reflected in formal logic). 
Note: though the paper utilises many Tich\'{y}'s ideas, it also employs numerous ideas developed by the present author, some of them being alien or even contradictory to Tich\'{y}'s.

\section{Partial type theory $\mathsf{TT^*}$}

We adopt here Tich\'{y}'s \cite{tichy1982,tichy1988} (see also Moschovakis \cite{moschovakis2006}) idea that expressions of language \textit{express} (or: depict) abstract, structured, not necessarily effective, acyclic \textit{algorithmic computations}, called by Tich\'{y} \textit{constructions}. 
In our construal \cite{raclavsky2020,kuchynka-raclavsky2024},
constructions construct objects -- each from a particular \textit{domain} 
 $\mathscr{D}_{\tau^n}$ that interprets the \textit{type} $\tau^n$ (see below)
 -- that are different from them. 
For an illustrative example, ``$3 \times 1$'' and ``$5 - 2$'' express two different (but congruent) constructions, namely $\acqs{\times} (\acqb 3 ,\acqb 1)$ and $\acqs{-}(\acqb 5,\acqb 2)$, of the number $\obj{3}$.
Constructions may aptly serve as \textit{fine-grained meanings} of expressions, while the objects constructed by them serve as their \textit{denotata} (the double-layered semantics is \textit{neo-Fregean} in its spirit):

$$expressions/terms 
\underset{express}\longrightarrow
constructions
\underset{construct}\longrightarrow
objects \, (denotata)
$$
\noindent

Constructing is dependent on \textit{assignment} $v$ and \textit{model} $\mathscr{M}$ (see 
below), 
so constructions are said to $v$-\textit{construct} objects \textit{in} $\mathscr{M}$. 
Constructions $v$-constructing other constructions in $\mathscr{M}$ are also allowed. 
Each \textit{assignment} $v$ (into \textit{frame} $\mathscr{F} \in \mathscr{M}$, see below)
is the union of all \textit{total} functions $v_i^{\tau^n}$, one for each ${\tau^n}$, such that each variable(-as-construction) $\mathsfit{x}_i$ \textit{ranging over} type $\tau^n$ is assigned a $\tau^n$-object $\mathtt{X}_i \in\mathscr{D}_{\tau^n}$.
Notation: $v(\vec{\mathtt{X}}/\vec{\mathsfit{x}})$ abbreviates 
$v(\mathtt{X}_1/\mathsfit{x}_1; ...; \mathtt{X}_m/\mathsfit{x}_m)$, which stands for $v$'s $\vec{\cc{x}}$-\textit{modification} $v'$ such that for each $1 \leq i \leq m$, it assigns 
a $\tau^n_i$-object $\mathtt{X}_i$ to $\mathsfit{x}_i/\tau_i^n$.

Some constructions, e.g. $\acqs \div (\acqb 3 ,\acqb 0)$, $v$-construct nothing at all in $\mathscr{M}$, they are called $v$-\textit{improper in} $\mathscr{M}$; they serve as meanings of \textit{non-denoting expressions}. 
Two constructions are called $v$-\textit{congruent in} $\mathscr{M}$, $\cong$, iff they $v$-construct the same object in $\mathscr{M}$ (examples above), or they are both $v$-improper in $\mathscr{M}$.

By \textit{functions} we mean here set-theoretical \textit{functions-as-mappings} (graphs, ...), not \textit{functions-as-computations}. 
Each function $\obj{f}$ has a certain domain $\mathscr{D}_\mathtt{x}$ of $\obj{f}$'s arguments and a (co-)domain $\mathscr{D}_\mathtt{y}$ of $\obj{f}$'s values; 
a function $\mathtt{f}$ is called \textit{total} / \textit{partial} iff all / some-but-not-all members of its $\mathscr{D}_\mathtt{x}$ are mapped to some members of its $\mathscr{D}_\mathtt{y}$. 
{Unlike} any total function, each partial function thus \textit{lacks a value} -- i.e. it's \textit{undefined} -- for at least one of its arguments.
Functions-as-computations may be identified with some constructions; some of them are \textit{strict}, 
so the applications involving them 
are $v$-improper in $\mathscr{M}$.

\subsection{Language $\mathscr{L}_\mathsf{TT^*}$}

Constructions are best recorded using $\lambda$-notation. Let for any $E_i$ (construction/object/type), $1 \leq i \leq m$, 
$\vec{E}$ be short for $E_1,..., E_m$, 
while ``$\lambda \vec{\cc{x}}.$'' rather unpacks to 
``$\lambda \cc{x}_1...\cc{x}_m.$''.
Whenever possible, we employ two languages: 
(i) an \textit{object language} whose part is e.g. ``$\obj{X}$'', which stands for the object $\obj{X}$ (which is often an object that isn't a construction) and 
(ii) a \textit{meta-language} whose part is e.g. 
``$\cc{X}$'', which stands for the construction $\cc{X}$ of $\obj{X}$ (if any). 
Let $X:=Y$ mean that $X$ is defined (takes the form, ...) as $Y$.

Each construction of $\mathsf{TT^*}$ (and so each $\mathscr{L}_\mathsf{TT^*}$'s proper expression) is always typed:

\begin{definition}[{Forms of constructions (and of terms of the language $\mathscr{L}_\mathsf{TT^*}$)}]
$\newline$

\vspace{-15pt}
\begin{center}

\begin{tabular}{p{15pt}lll}
&
\textit{Form of} $\cc{X}$:  \, & \textit{$\cc X$'s name:} & \textit{$\cc X$'s typing rule $\cc{X}/\tau^n$:}\\
\hline
\rm{i.} & 
$\cc{x}$ & \textit{variable}  & $\cc{x}/\tau^n$
\\
\rm{ii.} & 
$\acq{\obj{X}}$ & \textit{acquisition} & $\acq{\obj{X}}/\tau^n$;
if $\acq{\obj{X}}/\tau^n:\not = *^n$, one writes
$\acqb{X}$
\\
\rm{iii.} & 
$\cc{F}(\vec{\cc{X}})$ & \textit{application}  & $\cc F(\vec{\cc{X}})/\tau$, where $\cc{X}_1/\tau^n_1; ... ;\cc{X}_m/\tau^n_m;
\cc F/ \langle \vec{\tau}^n \rangle {\to} \tau^n
$
\\
\rm{iv.} & 
$\lambda \vec{\cc{x}} . \cc{Y}$ & $\lambda$-\textit{abstraction} & $\lambda \vec{\cc{x}} . \cc{Y}/\langle \vec{\tau}^n \rangle {\to} \tau^n$, where $\cc{Y}/\tau^n; \cc{x}_1/\tau^n_1; ...; \cc{x}_m/\tau^n_m$
\\
\end{tabular}
\end{center}

\end{definition}

\begin{small}
\noindent
\textit{Notes.}
Auxiliary expressions (note that we are not pedantic as regards quotation marks): $(,), \lambda \cc{x}.$ and $\ulcorner, \urcorner$; auxiliary brackets: $[,]$.
\textit{Acquisitions} $\acq{\obj{X}}$ are \textit{primitive} constructions, they are not applications of a certain function to $\obj{X}$. Each acquisition $\acq{\obj{X}}$ $v$-constructs $\obj{X}$ in just one direct construction step of `delivering' $\obj{X}$ and leaving it as it is.
Acquisitions can be thus seen as `procedural constants'; 
\textit{variables} are `procedural', too. 
\textit{Applications} $\cc{F}(\vec{\cc{X}})$ are `juxtapositions' of constructions such that if $\cc{F}$ $v$-constructs a function $\obj{f}$ in $\mathscr{M}$ whose argument $\langle \vec{\obj{x}}\rangle $ consists of entities 
$v$-constructed by $\vec {\cc{X}}$ in $\mathscr{M}$, and $\obj{f}$ is defined for 
$\langle \vec{\obj{x}}\rangle $, then the whole application $v$-constructs $\obj{y} := \obj{f}(\vec{\obj{x}})$ in $\mathscr{M}$. 
(Irreducibility of $m$-ary partial functions to unary ones, proved in \cite{tichy1982}, necessitates $\cc{F}(\vec{\cc{X}})$ instead of  $\cc{F}' (\cc{X}_{1} ... (\cc{X}_{m-1} (\cc{X}_{m})))$; similarly for types.)
Each \textit{abstraction} $\lambda \vec{\cc{x}} . \cc{Y}$ 
$v$-constructs a function $\obj{f}$ in $\mathscr{M}$ from $m$-tuples $v$-constructed by $\vec{\cc{X}}$ in $\mathscr{M}$ even on $\vec{\cc{x}}$-modifications of $v$, i.e. $v'$, 
to values that are $v^{(')}$-constructed in $\mathscr{M}$ by  abstraction's body $\cc{Y}$. See our \cite{kuchynka-raclavsky2024} for an \textit{exact} description of $\mathscr{L}_{\mathsf{TT^*}}$'s semantics.
\end{small}

\subsection{Types, orders, frames, models}

\textit{Typing}. 
Let $\tau^n, \tau^n_0, \vec{\tau}^n$ be \textit{type variables} (in the following sections, ``$^n$'' will be suppressed) and\linebreak
$o, \iota,*^1,  ...,*^n$ be \textit{type constants}.
Expressions of $\mathscr{L}_\mathsf{TT^*}$, but primarily 
$\mathsf{TT^*}$'s constructions, are typed via \textit{typing statements} of the form $\cc{X}/\tau^n,$ saying that for any $v$, the construction $\cc{X}$ should $v$-construct an object of type $\tau^n$; $\cc{X}/\tau^n$ is often called a $\tau^n$-\textit{construction}.
Notation: $\cc{X},\cc{Y}/\tau^n$ is short for $\cc{X}/\tau^n; \cc{Y}/\tau^n$.
Examples: $\cc{x}/ \tau^n$; 
$\acqs{\div} (\acqb{3},\acqb{1}), \acqs{\div} (\acqb{3}, \acqb{0})/\iota$, where $\iota$ is interpreted as $\mathbb{R}$; $\acqb{0}, \acqb{1}, \acqb{3}/\iota$;
 $\acqs{\div}/\langle \iota, \iota\rangle{\to}\iota$ (cf. below).

\textit{Interpretation of types.} 
Types $\tau^n$ are interpreted by sets of objects called \textit{domains} $\mathscr{D}_{\tau^n}$. 
Members of $\mathscr{D}_{\tau^n}$ are called $\tau^n$-\textit{objects}.
Let $\mathscr{T}$ be a set of types for $\mathscr{L}_\mathsf{TT^*}$.
A \textit{frame} $\mathscr{F} = \{ \mathscr{D}_{\tau^n} \,|\, {\tau^n} \in \mathscr{T} \}$ consists of all domains that interpret all types in $\mathscr{T}$; each $\mathscr{D}_{\tau^n} \in \mathscr{F}$ contains the equality relation $=^{\tau^n}$ and $\Sigma^{\tau^n}$ (below). 
A \textit{model} $\mathscr{M}$ is an \textit{interpretation} for $\mathscr{L}_\mathsf{TT^*}$, i.e. a couple $\langle \mathscr{F} , \mathscr{I}\rangle $ such that the \textit{interpretation mapping} $\mathscr{I}$ maps acquisitions expressed by $\mathscr{L}_\mathsf{TT^*}$'s constants (e.g. ``$=^{\tau^n}$'') to objects of $\mathscr{F}$ (\cite{kuchynka-raclavsky2024}).

\begin{definition}[{Types $\tau^n$}]
Let $1 \leq n  \in \mathbb{N}$. 
\begin{itemize}
\setlength\itemsep{0.1em}
\item[$\mathscr{B}$]
Let $\mathscr{B}= \{ o, \iota \}$ be a \emph{type base} for $\mathscr{L}_\mathsf{TT^*}$ 
such that $\mathscr{D}_o=\{ \obj{T}, \obj{F}\}$ (\textit{truth values}; $\obj{T} \not= \obj{F}$) and 
$\mathscr{D}_\iota$ are `\textit{entities}' (e.g. $\mathscr{D}_\iota=\mathbb{R}$).

\item[$\tau^1$]
\emph{$1$st-order types}:
(a) each type $\tau^\mathscr{B} \in \mathscr{B}$ is a $1$\emph{st-order type}, and 
(b) if $\vec{\tau}^1$ and $\tau_0^1$ are $1$st-order types, $\langle \vec{\tau}^1 \rangle {\to} \tau^1_0$ is also a $1$\emph{st-order type}; $\mathscr{D}_{\langle \vec{\tau}^1 \rangle {\to} \tau^1_0}$ consists of total and partial functions ${\mathscr{D}}_{{\tau}^1_1} \times ... \times {\mathscr{D}}_{{\tau}^1_m}  \to \mathscr{D}_{{\tau}^1_0}$. 

\item[$*^n$]
Let $*^n$ be type such that $\mathscr{D}_{*^n}$ consists of all $n$\emph{th-order constructions}, i.e. constructions whose 
subconstructions $v$-construct (if $v$-proper) objects in $\mathscr{M}$ of $n$th-order types. 

\item[$\tau^{n+1}$]
$(n{+}1)$\emph{st-order types}: 
(a) each $n$th-order type $\tau ^n$ is an $(n{+}1)$\emph{st-order type}; 
(b) the type $*^n$ is an $(n{+}1)$\emph{st-order type},
and, 
(c) if $\vec{\tau}^{n{+}1}$ and $\tau^{n{+}1}_0$ are $(n+1)$st-order types, 
then $\langle \vec{\tau}^{n{+}1} \rangle {\to} \tau^{n{+}1}_0$ is also an $(n{+}1)$\emph{st-order type};
$\mathscr{D}_{\langle \vec{\tau}^{n{+}1} \rangle {\to} \tau^{n{+}1}_0}$ consists of total and partial functions ${\mathscr{D}}_{{\tau}^{n{+}1}_1} \times ... \times {\mathscr{D}}_{{\tau}^{n{+}1}_m}  \to \mathscr{D}_{{\tau}^{n{+}1}_0}$. 
\end{itemize}
\end{definition}

\begin{small}
\noindent
\textit{Notes.}
Auxiliary brackets: $(, )$.
Types defined in steps ($\tau^1$.b) and ($\tau^{n+1}$.c) are called \textit{function types}, for they're interpreted by domains consisting of $m$-ary functions.
\textit{Sets} of $\tau$-objects, i.e. of $\mathscr{D}_\tau$'s members, are identified with characteristic functions in $\mathscr{D}_{\tau{\to}o}$; similarly for $m$-ary \textit{relations}.
Domains are pairwise disjoint, except 
$\mathscr{D}_{*^1} \subset \mathscr{D}_{*^2} \subset ... \subset 
\mathscr{D}_{*^n} $ (\textit{cumulativity});
there is no greatest order $n \in \mathbb{N}$. Neither \textit{Russell's paradox}, nor e.g. \textit{Russell-Myhill's paradox} about propositions (as identified with $o$-constructions) is possible in $\mathsf{TT^*}$ (cf. \cite{raclavsky2020}).
We cannot enjoy the higher orders in this short paper.
If $\cc{X}$ $v$-constructs $\obj{X}$ (if any) in $\mathscr{M}$:
$\cc{X}/\tau^n$ indicates 
$\obj{X}\in \mathscr{D}_{\tau^n}$ and $\cc{X} \in \mathscr{D}_{*^n}$. 
\end{small}

\section{Natural deduction in sequent style, $\mathsf{ND_{TT^*}}$}

$\mathsf{ND_{TT^*}}$, which we borrow and slightly adjust  from \cite{kuchynka-raclavsky2024,raclavsky2022},
stems from Tich\' {y}'s systems \cite{tichy1982,tichy1986} for his  partial $\mathsf{TT}$. It's essentially an $\mathsf{ND}$ \textit{in sequent style}, but with `\textit{signed  formulas}', 
so it's a kind of \textit{labelled calculi}, cf. Gabbay \cite{gabbay1996}.
In Kuchy\v{n}ka and Raclavsk\'{y} \cite{kuchynka-raclavsky2024}, \textit{Henkin-completeness} of $\mathsf{ND_{TT^*}}$, and thus the \textit{higher-order logic} ($\mathsf{HOL}$) we apply here, w.r.t. an exact semantics of $\mathscr{L}_{\mathsf{TT^*}}$ is proved in details.

\subsection{Matches, sequents and derivation rules}

$\mathsf{ND}_\mathsf{TT*}$'s 
\textit{rules} $\mathtt{R}$ are made from sequents, while 
\textit{sequents} $\mathtt{S}$ are made from 
$\mathsf{ND}_\mathsf{TT*}$'s \textit{statements} called \textit{matches} $\mathtt{M}$. Here are three motivations a.--c. for introducing matches.

a. \qquad 
Each $\mathtt{M}$ states $v$-\textit{congruence} in $\mathscr{M}$ of a certain (typically compound) construction $\cc{X}$ with a (typically simple) variable or acquisition $\acqb{x}$. So the best notation for $\mathtt{M}$ would be $\cc{X} \cong \acqb{x}$, where $\cong $ is the \textit{strong equality} operator 
(it holds even if $\cc{X}$ and $\acqb{x}$ are both $v$-improper $\mathscr{M}$), which indicates the `equational character' of the system.
We rather write $\cc{X} {\,:^\tau\,} \acqb{x}$, which displays the type $\tau$ of each of $\cc{X}$ and $\acqb{x}$
and underlines that matches present \textit{signed formulas}.
As signed formulas, matches obviously increase the deduction power of $\mathsf{ND_{TT^*}}$; to illustrate, from 
$\acqs{\supset} (\varphi,\psi) {\,:^o\,}  \acqb{F}$ one deduces e.g. $\psi {\,:^o\,}  \acqb{F}$. The term ``match'' is of course auxiliary and our above explanation admittedly specific: ``$\varphi \; \textrm{true}$'' or ``$\textrm{T}: \varphi$'' (both saying `the formula $\varphi$ has the value True', which is encoded even by our $\varphi {\,:^o\,}  \acqb{T}$) are a familiar and ubiquitous concept in most (if not all) computer-science-related writings on natural deduction and was first employed in semantic tableaux method.

b. \qquad 
The use of signed formulas is especially fruitful when 
dealing with partiality. 
Let $\acqs{\bot}^\tau/\tau$ be any $v$-improper $\tau$-construction; ``$^\tau$'' will usually be suppressed. 
$\acqs{\bot}^\tau$ may perhaps seem to play a role of so-called 
\textit{dummy value} (or \textit{null value}) known from algebraic approaches of e.g.
$\mathsf{FL}$ by Scott \cite{scott1979}. 
But there is a crucial difference:
Scott and many others use \textit{denotational semantics} in which something (namely the dummy value) must interpret a non-denoting expression, otherwise it's meaningless (just as non-well-formed expressions); in the \textit{procedural semantics} followed in this paper, however, a non-denoting (well-formed) expression lacks denotation (reference), but expresses as its \textit{meaning} a specific improper construction $\acqs{\bot}^\tau$.
To illustrate such matches, let $\acqb{3}, \acqb{0} / \iota$ (the numbers-as-objects $3,0$), $ \acqs{\div} / \langle \iota, \iota \rangle {\to} \iota$ (the familiar division mapping): the match 
$\acqs{\div}  (\acqb{3}, \acqb{0}) :^\iota \acqs{\bot}$ says that the two constructions flanking $:^\iota$ are $v$-congruent in $\mathscr{M}$
(for they are both $v$-improper); note that we do not postulate
a `dummy number' in our ontology that is allegedly computed by $\acqs{\div} (\acqb{3}, \acqb{0})$.

c. \qquad 
Last but not least, the \textit{monotonicity} of 
$\vDash$ is preserved, for 
each $\mathtt{M}$ definitely either holds, or not. 
Then the following situation of common partial logics, criticised by Blamey \cite{blamey1986}, is excluded:  
let $\sim$ be the familiar function of negation; 
if $\varphi$ and so even $\sim \varphi$ have the value $\bot$, and 
$\varphi \vDash \psi $, then 
$\sim \psi \nvDash \sim \varphi $.

\paragraph{i. Matches.} Let $\cc{X}, \cc{x}, \acqb{X}, \acq{\obj{X}} /\tau$.
Matches split into two types, a. and b.
Each of three a.-type matches 
\[
\text{ a. } 
\qquad
\mathtt{M} := 
 \cc{X} {\; :^{\tau} \;}\acqb{X}
 \;|\;
\cc{X} {\; :^{\tau} \;}\acq {\obj{X}}
 \;|\;
\cc{X} {\; :^{\tau} \;}\cc{x}
\]
says that $\cc{X}$ is $v$-proper in $\mathscr{M}$. Notation: $\cc{X}{\; :^{\tau} \;}\acqb{x}$ represents any a.-type matches. 
Each b.-type match

\[
\text{ b. } 
\qquad 
\mathtt{M} := 
\cc{X} {\; :^{\tau} \;}\bot
\]
says that  $\cc{X}$ is $v$-improper in $\mathscr{M}$. 
Notation: $\cc{X}{\; :^{\tau} \;}\underline{\acqb{x}} $ covers variants $\cc{X}{\; :^{\tau} \;}\acqb{x}$ and $\cc{X}{\; :^{\tau}\;}\acqs{\bot}$. 
An assignment $v$ \textit{satisfies} $\cc{X}{:^\tau \;} \underline{\acqb{x}}$ in $\mathscr{M}$ iff $\cc{X} \cong \underline{\acqb{x}}$ in $\mathscr{M}$.

\paragraph{ii. Sequents.} 
A \textit{sequent} 
\[
\mathtt{S}:= \; \Gamma \longrightarrow \mathtt{M}
\]
may be seen as a couple consisting of a finite \textit{set} (not multiset) $\Gamma$ 
of matches and a match $\mathtt{M}$ 
that follows from $\Gamma$.
$\mathtt{S}$ is \textit{valid} in $\mathscr{M}$ iff every $v$ that satisfies all members of $\Gamma$ in $\mathscr{M}$ also satisfies $\mathtt{M}$ in $\mathscr{M}$. Notation: where $\Delta $ is a set of matches,
$\Gamma , \Delta  \longrightarrow \mathtt{M} $ abbreviates
$\Gamma \cup \Delta  \longrightarrow \mathtt{M} $; 
$\Gamma , \mathtt{M} \longrightarrow \mathtt{M} $ abbreviates
$\Gamma \cup \{ \mathtt{M} \} \longrightarrow \mathtt{M}$.

\paragraph{iii. Rules.} 
A \textit{(derivation) rule} 
$\vec{\mathtt{S}} \vdash \mathtt{S}$, is a validity-preserving operation on sequents, usually written 
\vspace{-0pt}
\begin{prooftree}
\AxiomC{$\vec{\mathtt{S}}$}
\LeftLabel{$\mathtt{R} := \;$}
\RightLabel{,}
\UnaryInfC{$\mathtt{S}$}
\end{prooftree}

\vspace{-0pt}
\noindent
where $\vec{\mathtt{S}}$ are its \textit{premisses}, 
$\mathtt{S}$ its \textit{conclusion}.
Each $\mathtt{R}$ says that $\mathtt{S}$ is valid in all models 
in which $\vec{\mathtt{S}}$ are valid.

Let $H$ be an arbitrary \textit{set of sequents}. 
A finite \textit{sequence} $S$ \textit{of sequents}, each member of which being either 
a member of $H$, or the result of the application of a rule from a \textit{set of rules} $R$ to some preceding members of $S$ or members of $H$ is called a \textit{derivation $\mathtt{D}$} 
\textit{of} $S$'s last sequent $\mathtt{S}$ from \textit{H}. 
$\mathtt{D}$ is also called in brief \textit{proof} and (numbered) members of $S$ are called  \textit{steps}.
$H \vdash \mathtt{S}$ presents a \textit{derived rule}.

\subsection{$\mathsf{ND}_{\mathsf{TT}^*}$'s derivation rules}

The rules of $\mathsf{ND}_\mathsf{TT^*}$ may be divided into four groups: i. \textit{structural rules}, ii. \textit{form rules}, iii. \textit{operational rules} 
and iv. \textit{rules for extralogical constants}.
The i.-type rules present general properties of validity, 
the ii.-type rules present properties of validity w.r.t. forms of constructions.
The iii.-type rules make $\mathsf{TT}^*$ a $\mathsf{HOL}$.

Even a cursory inspection of the i.- and ii.-type rules  
reveals that they rather resemble rules familiar from $\mathsf{ND}$ for modern $\mathsf{STT}$, compare e.g. Hindley and Seldin \cite{hindley-seldin2008} and $\mathsf{ND_{TT^*}}$'s rules (AX), (WR), (CUT) (see Def. 3 below). 
Those $\textsf{ND}$s usually utilise
sequents of the form $\Gamma \longrightarrow t : \tau$, in which term $t$ is typed by $\tau$, while we use 
$\Gamma \longrightarrow  \cc{X} {\; :}^{\tau}\, \underline{\acqb{x}}$ to the same effect.
Nevertheless, labelling $\cc{X}$ by $\underline{\acqb{x}}$ (cf. below) for the reasons stated above gives rise to a few new 
rules; in Def. 3, see esp. (EXH).
Deduction systems  $\mathsf{STT}$ by 
Beeson \cite{beeson1985}, Feferman \cite{feferman1995} and 
Farmer \cite{farmer1990} are not sequent-style ones as $\mathsf{ND_{TT^*}}$ is, 
so their encoding mechanisms differ.
To illustrate, the fact that both variables and constants always denote is expressed by 
their axioms $x_\tau \downarrow$ (where $\downarrow$ reads `is denoting') and 
$c_\tau \downarrow$, while $\mathsf{ND_{TT^*}}$ uses (TM) (cf. Def. 4) for both; similarly for $\lambda x_\tau. t \downarrow$ and our ($\lambda$-INST) (cf. Def. 4).

Notational agreement (holding unless stated otherwise).
Let the following symbols be any:   
$\mathtt{M}_{(i)}$ -- match; $\mathtt{S}_{(i)}$ -- sequent; 
$\Gamma$ (or $\Delta$) -- set of matches;  
$\cc{x}_{(i)}, \cc{y}, \cc f, \cc g$ -- variables; $\acqb{x}_{(i)}, \acqb{y}, \acqb{f}, \acqb{g}$ --  
acquisitions/variables; 
$\cc{X}_{(i)},\cc{Y}, \cc F$ -- 
constructions.
The constructions fit types as follows:
 $\cc{X}, \cc{Y}, \acqb{x}, \acqb{y}/\tau; \acqb{x}_1, \cc{X}_1/
\tau; ...; \acqb{x}_m, \cc{X}_m/\tau; \cc F,  $ 
\\ 
$\acqb{f}, \acqb{g}/ \langle \vec{\tau} \rangle {\to} \tau$; let $\phi$ abbreviate $\langle \vec{\tau} \rangle  {\to} \tau$. 
\textit{Conditions} of each relevant $\mathtt{R}$ typically include: (i) the variables occurring within $\mathtt{R}$ are pairwise distinct and (ii) they are not free in $\Gamma, \mathtt{M}$ and other constructions occurring in $\mathtt{R}$.\footnote{$\cc{x}$ is called \textit{free in} $\mathtt{M}$ of the form $\cc{X}{\; :^{\tau} \;}\underline{\acqb{x}}$ iff it's free in $\cc{X}$ or $\acqb{x}$; $\cc{x}$ is called \textit{free in} $\Gamma$ iff it's free at least in one $\mathtt{M}_{i} \in \Gamma$.}
Let $\cc{Y}_{[\cc{X}/\cc{x}]}$
stand for the construction $\cc{Y}$ in which \textit{free} occurrences of $\cc x$ are \textit{substituted} by $\cc X$, as defined in \cite{kuchynka-raclavsky2024};
steps differing by rewriting terms on the basis of substitution are suppressed.

\medskip

\begin{definition}[Structural rules]
For informal description of (nearly all) our rules, see \cite{raclavsky2020}. (AX) is the \textit{axiom rule};  (WR) is the \textit{weakening rule}; (CUT) is the \textit{deletional cut rule} (cf. \cite{gabbay1996}); (EFQ) is the \textit{ex falso/contradictione quodlibet} \textit{rule}. (EXH) is the \textit{exhaustation rule} -- it says that if the assumptions that $\cc{X}$ is / is not $v$-proper are needed for 
$\mathtt{M}$'s following from $\Gamma$, then $\mathtt{M}$ follows from $\Gamma$ independently of the assumptions.
\begin{small}
$\newline$

\vspace{-10pt}
\begin{center}
\begin{minipage}{0.30\textwidth}

\begin{prooftree}
\AxiomC{}
\RightLabel{\rm (AX)}
\UnaryInfC{$\Gamma, \mathtt{M} \longrightarrow \mathtt{M}$}
\end{prooftree}

\end{minipage}
\begin{minipage}{0.38\textwidth}

\begin{prooftree}
\AxiomC{$\Gamma \longrightarrow \mathtt{M}_1$}
\AxiomC{$\Gamma,\mathtt{M}_1 \longrightarrow \mathtt{M}_2$}
\RightLabel{\rm (CUT)}
\BinaryInfC{$\Gamma \longrightarrow \mathtt{M}_2$}
\end{prooftree}

\end{minipage}
\begin{minipage}{0.30\textwidth}

\begin{prooftree}
\AxiomC{$ \Gamma \longrightarrow \mathtt{M}$}
\RightLabel{\rm (WR)}
\UnaryInfC{$\Gamma, \Delta \longrightarrow \mathtt{M}$}
\end{prooftree}

\end{minipage}

\smallskip

\begin{minipage}{0.48\textwidth}

\begin{prooftree}
\AxiomC{$ \Gamma \longrightarrow  \mathtt{M}_1$}
\hspace{-20pt}
\AxiomC{$ \Gamma \longrightarrow  \mathtt{M}_2$}
\RightLabel{\rm (EFQ)}
\BinaryInfC{$\Gamma \longrightarrow \mathtt{M}$}
\end{prooftree}

\end{minipage}
\begin{minipage}{0.50\textwidth}
\begin{prooftree}
\AxiomC{$\Gamma, \cc{X}{\; :^{\tau} \;}\acqs{\bot} \longrightarrow \mathtt{M}$}
\hspace{-20pt}
\AxiomC{$\Gamma, \cc{X} {\; :^{\tau} \;}\cc{x} \longrightarrow \mathtt{M} $}
\RightLabel{\rm (EXH)}
\BinaryInfC{$\Gamma \longrightarrow \mathtt{M}$
}
\end{prooftree}
\end{minipage}

\end{center}

\smallskip
\noindent
{\footnotesize 
\textit{Condition} {\rm (EFQ)}: $\mathtt{M}_1$ and $\mathtt{M}_2$ 
are \emph{patently incompatible} -- they are either of the forms $\cc{X}{\; :^{\tau} \;}\acqb{x}$ and $\cc{X}{\; :^{\tau} \;}\acqs{\bot}$, or of the forms
$\cc{X}{\; :^{\tau} \;}\acqb{x}_1$ and $\cc{X} {\; :^{\tau} \;}\acqb{x}_2$, where $\acqb{x}_1$ and $\acqb{x}_2$ acquire distinct objects $\obj{X}_1$ and $\obj{X}_2$.
Patently incompatible matches are never satisfied (in $\mathscr{M}$) by the same $v$.}

\end{small}
\end{definition}

\begin{definition}[Form rules]
The sense of (TM), the \textit{trivial match rule}, and
($\lambda$-INST), the $\lambda$-\textit{instantiation rule}, was indicated above. The rest of the form rules govern applications that are $v$-proper. The $\beta$-\textit{conversion rules} ($\beta$-CON) (\textit{contraction r.}) and ($\beta$-EXP) (\textit{expansion r.}) are very important, while the rules for \textit{substitution in applications} (a-SUB) are very useful, too; (EXT)  is the \textit{extensionality rule}.

\begin{small}
$\newline$

\vspace{-15pt}
\begin{center}

\begin{minipage}{0.30\textwidth}
\begin{prooftree}
\AxiomC{}
\RightLabel{\rm (TM)}
\UnaryInfC{$\Gamma \longrightarrow \acqb{x} {\; :^{\tau} \;}\acqb{x} $}
\end{prooftree}
\end{minipage}
\begin{minipage}{0.33\textwidth}

\begin{prooftree}
\AxiomC{$ \Gamma, \lambda \vec{\cc{x}} . \cc{Y} {\; :^{\phi} \;} \cc f \longrightarrow \mathtt{M}$}
\RightLabel{\rm ($\lambda$-INST)}
\UnaryInfC{$\Gamma \longrightarrow \mathtt{M}$}
\end{prooftree}
\end{minipage}
\begin{minipage}{0.35\textwidth}
\begin{prooftree}
\AxiomC{$\Gamma \longrightarrow [\lambda \vec{\cc{x}} . \cc{Y}] (\vec{\cc{X}}) {\; :^{\tau} \;}\acqb{y} $}
\RightLabel{\rm ($\beta$-CON)}
\UnaryInfC{$\Gamma \longrightarrow  \cc{Y}_{(\vec{\cc{X}}/\vec{\cc{x}})} {\; :^{\tau} \;} \acqb{y}$}
\end{prooftree}
\end{minipage}

\begin{prooftree}
\AxiomC{$\Gamma \longrightarrow  \cc{Y}_{(\vec{\cc{X}}/\vec{\cc{x}})} {\; :^{\tau} \;} {\acqb{y}}$}
\AxiomC{$\Gamma \longrightarrow \cc{X}_1 {\; :}^{\tau_1}  \acqb{x}_1 $}
\AxiomC{$... $}
\AxiomC{$\Gamma \longrightarrow \cc{X}_m {\; :^{\tau} \;} \acqb{x}_m$}
\RightLabel{\rm ($\beta$-EXP)}
\QuaternaryInfC{$\Gamma \longrightarrow [\lambda \vec{\cc{x}} . \cc{Y}] (\vec{\cc{X}}) {\; :^{\tau} \;} {\acqb{y}}$}
\end{prooftree}

\begin{prooftree}
\AxiomC{$\Gamma \longrightarrow  \cc F(\vec{\cc{X}}){\; :^{\tau} \;} \acqb{y}$}
\AxiomC{$\Gamma \longrightarrow  \cc{X}_1{\; :^{\tau_1}  \;} \acqb{x}_1$}
\AxiomC{$...$}
\AxiomC{$\Gamma \longrightarrow  \cc{X}_m {\; :^{\tau_m} \;}  \acqb{x}_m$}
\RightLabel{\rm (a-SUB.i)}
\QuaternaryInfC{$\Gamma \longrightarrow \cc{F}(\vec{\acqb{x}}){\; :^{\tau} \;} \acqb{y}$}
\end{prooftree}

\begin{prooftree}
\AxiomC{$ \Gamma \longrightarrow \cc F(\vec{\acqb{x}}){\; :^{\tau} \;} \acqb{y}$}
\AxiomC{$ \Gamma \longrightarrow \cc{X}_1 {\; :^{\tau_1} \;}  \acqb{x}_1 $}
\AxiomC{$...$}
\AxiomC{$\Gamma \longrightarrow  \cc{X}_m {\; :^{\tau_m} \;}  \acqb{x}_m $}
\RightLabel{\rm (a-SUB.ii)}
\QuaternaryInfC{$\Gamma \longrightarrow \cc F(\vec{\cc{X}}) {\; :^{\tau} \;} \acqb{y}$}
\end{prooftree}

\begin{prooftree}
\AxiomC{$\Gamma \longrightarrow  \cc F(\vec{\cc{X}}){\; :^{\tau} \;} \acqb{y}$}
\AxiomC{$\Gamma , \cc F{\; :^{\phi} \;}  \cc f, \cc{X}_1{\; :}^{\tau_1}  \cc{x}_1, ..., \cc{X}_m {\; :^{\tau_m} \;} \cc{x}_m \longrightarrow  \mathtt{M}$}
\RightLabel{\rm (a-INST)}
\BinaryInfC{$\Gamma \longrightarrow \mathtt{M}$}
\end{prooftree}

\noindent
\begin{minipage}{1.00\textwidth}
\begin{prooftree}
\AxiomC{$ \Gamma , \acqb{f}(\vec{\cc{x}}) {\; :^{\tau} \;} \cc{y} \longrightarrow  \acqb{g}(\vec{\cc{x}}){\; :^{\tau} \;} \cc{y}$}
\AxiomC{$ \Gamma , \acqb{g}(\vec{\cc{x}}) {\; :^{\tau} \;}\cc{y} \longrightarrow  \acqb{f}(\vec{\cc{x}}) {\; :^{\tau} \;} \cc{y}$}
\RightLabel{\rm (EXT)}
\BinaryInfC{$\Gamma \longrightarrow \acqb{g}{\; :^{\tau} \;}\acqb{f}$}
\end{prooftree}
\end{minipage}

\begin{prooftree}
\AxiomC{$\Gamma \longrightarrow \mathsfit{F} {:^{\phi} \;} \mathsfit{f}$}
\AxiomC{$\Gamma , \mathsfit{X}_1 {\; :^{\tau_1}\;} \mathsfit{x}_1; ... ; \mathsfit{X}_m {\; :^{\tau_m}\;} \mathsfit{x}_m \longrightarrow \mathtt{M}$}
\RightLabel{\rm (a-IMP$^\bot$)
\qquad
{\footnotesize 
Condition: 
except $\tau$, $\phi :\neq \langle \vec{\tau}\rangle {\to}\tau$.
}
}
\BinaryInfC{$\Gamma \longrightarrow \mathsfit{F}(\vec{\mathsfit{X}}) {\; :^{\tau} \;} \acqs{\bot}$}
\end{prooftree}

\end{center}
\end{small}
\end{definition}

$\mathsf{TT^*}$ employs the following familiar functions-as-mappings:
the \textit{negation} $\sim$ maps $\mathtt{T}$ to $\mathtt{F}$ and vice versa; 
the \textit{material conditional} $\supset$ maps
$\langle \mathtt{T}, \mathtt{F} \rangle$ to $\mathtt{F}$ but
$\langle \mathtt{T}, \mathtt{T} \rangle$,
$\langle \mathtt{F}, \mathtt{T} \rangle$,
$\langle \mathtt{F}, \mathtt{F} \rangle$ to $\mathtt{F}$; 
the \textit{universal quantifier} $\Pi^{\tau}$
maps the function $\mathscr{D}_{\tau} {\to} \mathscr{D}_o$ that assigns $\mathtt{T}$ to all $\tau$-objects to $\mathtt{T}$, but all other functions $\mathscr{D}_{\tau} {\to} \mathscr{D}_o$ to $\mathtt{F}$; 
the \textit{existential quantifier} $\Sigma^{\tau}$ (irreducible to $\Pi^{\tau}$, \cite{raclavsky2022}) maps each function $\mathscr{D}_{\tau} {\to} \mathscr{D}_o$ that assigns $\mathtt{T}$ to at least one $\tau$-object to $\mathtt{T}$, but all other functions $\mathscr{D}_{\tau} {\to} \mathscr{D}_o$ to $\mathtt{F}$; 
the \textit{identity relation} $=^{\tau}$ maps each couple pairing the same $\tau$-object to $\mathtt{T}$, but couples pairing different $\tau$-objects to 
$\mathtt{F}$;
the \textit{singularization} (or \textit{iota}) \textit{function}
$\rotatediota^\tau$ maps 
each function $\mathscr{D}_{\tau} {\to} \mathscr{D}_o$ that assigns $\mathtt{T}$ to just one $\tau$-object to that $\tau$-object, 
and is undefined for all other functions $\mathscr{D}_{\tau} {\to} \mathscr{D}_o$.
Their acquisitions $\acqs{\sim}, \acqs{\supset}, \acqs{\Pi}^\tau , \acqs{\Sigma}^\tau , \acqs{=}^\tau , \acqs{\rotatediota}^\tau$ are governed by the following rules. 

\medskip

\begin{definition}[Operational rules]
Specifying Def. 1 (point iii): $\acqb{T}, \acqb{F}, \acqb{o}, \acqb{o}', \cc{O}, \cc{O}' /o; 
\acqs{\sim} / o {\to} o; 
\acqs{\supset} / \langle o, o\rangle  {\to} o;$
$\acqs{\Pi}^{\tau} , \acqs{\Sigma}^{\tau} / (\tau {\to} o) {\to} o; 
\acqs{=}^{\tau} / \langle \tau , \tau \rangle {\to} o; 
\acqs{\rotatediota}^{\tau} / (\tau \to o ) {\to} \tau; \acqs{\bot}^\tau/ \tau;
\acqb{c}, \cc{C} / \tau{\to}o$ (`class').

\begin{center}

\begin{minipage}{.45\textwidth}
\begin{prooftree}
\AxiomC{$\Gamma , \acqb{o} {:^o \;}  \acqb{o}' \longrightarrow  \mathtt{M}_1$}
\AxiomC{$\Gamma , \acqb{o} {:^o \;}  \acqb{o}' \longrightarrow  \mathtt{M}_2$}
\RightLabel{($\sim$-I)}
\BinaryInfC{$\Gamma \longrightarrow  \acqs{\sim}(\acqb{o}) {:^o \;}  \acqb{o}'$}
\end{prooftree}
\end{minipage}
\begin{minipage}{.53\textwidth}
\textit{Condition} ($\sim$-I): $\mathtt{M}_1$ and $\mathtt{M}_2$ are patently incompatible.
\end{minipage}

\smallskip

\begin{minipage}{.53\textwidth}
\begin{prooftree}
\AxiomC{$\Gamma , \acqb{o} {:^o \;}  \acqb{T} \longrightarrow  \mathtt{M}$}
\AxiomC{$\Gamma , \acqb{o} {:^o \;}  \acqb{F} \longrightarrow  \mathtt{M}$}
\RightLabel{(RA)}
\BinaryInfC{$\Gamma \longrightarrow \mathtt{M}$}
\end{prooftree}
\end{minipage}
\begin{minipage}{.45\textwidth}
\begin{prooftree}
\AxiomC{$\Gamma , \acqs{\sim}(\acqb{o}) {:^o \;}  \cc{o} \longrightarrow  \mathtt{M} $}
\RightLabel{($\sim$-INST)}
\UnaryInfC{$\Gamma \longrightarrow  \mathtt{M} $}
\end{prooftree}
\end{minipage}

\medskip

\noindent
\begin{minipage}{.25\textwidth}
\begin{prooftree}
\AxiomC{$\Gamma , \acqb{o} {:^o \;} \acqb{T} 
\longrightarrow \acqb{o}' {:^o \;} \acqb{T} 
$}
\RightLabel{($\supset$-I)}
\UnaryInfC{$\Gamma \longrightarrow \acqs{\supset} (\acqb{o}, \acqb{o}' ) {:^o \;} \acqb{T} $}
\end{prooftree}
\end{minipage}
\begin{minipage}{.41\textwidth}
\begin{prooftree}
\AxiomC{$\Gamma \longrightarrow  \acqs{\supset} (\cc{O}, \cc{O}') {:^o \;} \acqb{T} $}
\AxiomC{\hspace{-17pt}$\Gamma \longrightarrow \cc{O} {:^o \;} \acqb{T} $}
\RightLabel{($\supset$-E)}
\BinaryInfC{$\Gamma \longrightarrow  \cc{O}' {:^o \;} \acqb{T} $}
\end{prooftree}
\end{minipage}
\begin{minipage}{.32\textwidth}
\begin{prooftree}
\AxiomC{$\Gamma , \acqs{\supset} (\acqb{o}, \acqb{o}') {:^o \;} \cc{o} \longrightarrow \mathtt{M}$}
\RightLabel{($\supset$-INST)}
\UnaryInfC{$\Gamma \longrightarrow  \mathtt{M} $}
\end{prooftree}
\end{minipage}

\medskip

\noindent
\begin{minipage}{.23\textwidth}
\begin{prooftree}
\AxiomC{$\Gamma \longrightarrow  \cc{C} (\cc{X}) {:^o \;} \acqb{T} $}
\RightLabel{($\Sigma$-I)}
\UnaryInfC{$\Gamma \longrightarrow \acqs{\Sigma}^\tau (\cc{C}) {:^o \;} \acqb{T}  $}
\end{prooftree}
\end{minipage}
\begin{minipage}{.42\textwidth}
\begin{prooftree}
\AxiomC{$\Gamma \longrightarrow  \acqs{\Sigma}^\tau (\cc{C}) {:^o \;} \acqb{T} $}
\AxiomC{\hspace{-8pt}$\Gamma , \cc{C} (\acqb{x}) {:^o \;} \acqb{T} \longrightarrow  \mathtt{M}$}
\RightLabel{($\Sigma$-E)}
\BinaryInfC{$\Gamma \longrightarrow  \mathtt{M}$}
\end{prooftree}
\end{minipage}
\begin{minipage}{.33\textwidth}
\begin{prooftree}
\AxiomC{$\Gamma , \acqs{\Sigma}^\tau (\acqb{c}) {:^o \;} \cc{o} \longrightarrow \mathtt{M}$}
\RightLabel{($\Sigma$-INST)}
\UnaryInfC{$\Gamma \longrightarrow  \mathtt{M} $}
\end{prooftree}
\end{minipage}

\medskip

\noindent
\begin{minipage}{.23\textwidth}
\begin{prooftree}
\AxiomC{$\Gamma \longrightarrow  \cc{C} (\cc{x}) {:^o \;} \acqb{T} $}
\RightLabel{($\Pi$-I)}
\UnaryInfC{$\Gamma \longrightarrow \acqs{\Pi}^\tau (\cc{C}) {:^o \;} \acqb{T}  $}
\end{prooftree}
\end{minipage}
\begin{minipage}{.36\textwidth}
\begin{prooftree}
\AxiomC{$\Gamma \longrightarrow  \acqs{\Pi}^\tau (\cc{C}) {:^o \;} \acqb{T} $}
\RightLabel{($\Pi$-E)}
\UnaryInfC{$\Gamma \longrightarrow  \cc{C} (\acqb{x}) {:^o \;} \acqb{T} $}
\end{prooftree}
\end{minipage}
\begin{minipage}{.32\textwidth}
\begin{prooftree}
\AxiomC{$\Gamma , \acqs{\Pi}^\tau (\acqb{c}) {:^o \;} \cc{o} \longrightarrow \mathtt{M}$}
\RightLabel{($\Pi$-INST)}
\UnaryInfC{$\Gamma \longrightarrow  \mathtt{M} $}
\end{prooftree}
\end{minipage}

\medskip

\noindent
\begin{minipage}{.25\textwidth}
\begin{prooftree}
\AxiomC{$\Gamma \longrightarrow \cc{X}  {:^\tau \;} \acqb{x} $}
\RightLabel{($=$-I)}
\UnaryInfC{$\Gamma \longrightarrow \acqs{=}^\tau (\cc{X}, \acqb{x} ) {:^o \;} \acqb{T}  $}
\end{prooftree}
\end{minipage}
\begin{minipage}{.33\textwidth}
\begin{prooftree}
\AxiomC{$\Gamma \longrightarrow  \acqs{=}^\tau (\cc{X}, \acqb{x}) {:^o \;} \acqb{T} $}
\RightLabel{($=$-E)}
\UnaryInfC{$\Gamma \longrightarrow   \cc{X}  {:^\tau \;} \acqb{x} $}
\end{prooftree}
\end{minipage}
\begin{minipage}{.25\textwidth}
\begin{prooftree}
\AxiomC{$\Gamma , \acqs{=}^\tau (\acqb{x}, \acqb{y}) {:^o \;} \cc{o} \longrightarrow \mathtt{M}$}
\RightLabel{($=$-INST)}
\UnaryInfC{$\Gamma \longrightarrow  \mathtt{M} $}
\end{prooftree}
\end{minipage}

\medskip

\noindent
\begin{minipage}{.47\textwidth}
\begin{prooftree}
\AxiomC{$\Gamma \longrightarrow   \cc{C} (\acqb{x}) {:^o \;} \acqb{T}$}
\AxiomC{$\Gamma ,  \cc{C} (\cc{y}) {:^o \;} \acqb{T} \longrightarrow  
\cc{y} {:^\tau \;} \acqb{x}$}
\RightLabel{($\rotatediota$-I)}
\BinaryInfC{$\Gamma \longrightarrow \acqs{\rotatediota}^\tau (\cc{C}) {:^\tau \;} \acqb{x}$}
\end{prooftree}
\end{minipage}
\begin{minipage}{.31\textwidth}
\begin{prooftree}
\AxiomC{$\Gamma \longrightarrow \acqs{\rotatediota}^\tau (\cc{C}) {:^\tau \;} \acqb{x}$}
\RightLabel{($\rotatediota$-E)}
\UnaryInfC{$\Gamma \longrightarrow   \cc{C} (\acqb{x}) {:^o \;} \acqb{T} $}
\end{prooftree}
\end{minipage}

\begin{minipage}{.3\textwidth}
\begin{prooftree}
\AxiomC{$\Gamma , \acqs{\rotatediota}^\tau (\cc{C}) {:^\tau \;} \acqb{x} \longrightarrow \mathtt{M} $}
\AxiomC{$\Gamma \longrightarrow   \cc{C} (\acqb{x}) {:^o \;} \cc{o} $}
\RightLabel{($\rotatediota$-INST)}
\BinaryInfC{$\Gamma \longrightarrow  \mathtt{M} $}
\end{prooftree}
\end{minipage}

\end{center}
\end{definition}

\textit{Notes on operational rules}.
There is a difference between (i) e.g. $\acqb{o}$, which is an acquisition or variable -- in both cases an always $v$-proper $o$-construction, and (ii) $\cc{O}$, which is any form of $o$-constructions -- which needn't be $v$-proper if an application occurs in the place of $\cc{O}$. Note then that all INST-rules require $\acqb{o}$ (etc.) being a $v$-proper construction (in systems that do not employ partial functions, or, more precisely, improper constructions, INST-rules are not needed). 
Rules such as ($\supset$-E) or ($\Pi$-E) omit the condition only seemingly: 
the  condition is imposed on $\cc{O}, \cc{O}'$ or $\cc{C}$ through the fact that $\acqs{\supset} (\cc{O}, \cc{O}')$ or $\acqs{\Pi}^\tau (\cc{C})$ are $v$-proper (they $v$-construct $\mathtt{T}$), hence their subconstructions $\cc{O}, \cc{O}'$ and $\cc{C}$ must be $v$-proper, too.
Most of the operational rules have a straightforward reading; for example, ($\supset$-E) says that if both an implication and its antecedent are true, then we may conclude that its consequent is also true; ($\Pi$-E) says that if the `higher-order concept' $\obj{All}$ applies to a set (some say: class) $\obj{C}$ of $\tau$-objects, then we may conclude that \textit{any} $\tau$-object $\obj{x}$ falls in $\obj{C}$. All these rules, incl. (RA), \textit{the redundant assumption rule}, occur in Tich\'{y}'s \cite{tichy1986} but without any comment or explanation.
But we add here $\iota$-rules (commented below), i.e. the rules for the \textit{iota operator} $\acqs{\rotatediota}^\tau$ first proposed by the present author in \cite{raclavsky2022}.

\textit{Notes on $\acqs{\rotatediota}^\tau$ rules}.
The $\rotatediota$-operator is \textit{prima facie} not `defined' in terms of $\exists, \forall, =$ as in standard approaches, cf. e.g. Russell's contextual introduction of $\rotatediota$-operator in \cite{whitehead-russell1910-13}, $G (\rotatediota x . F(x)) := \exists x ( (\forall y ( F(y) \leftrightarrow y=x )\land G(x) )$ (notation adjusted).
 But a conscientious eye quickly reveals that Russell's $y=x$ is encoded by our $\cc{y} {:^\tau} \cc{x}$ (reread our informal description of matches in 3.1.a). A version of the rule ($\iota$-I) with the match $\acqs{=}^\tau(\cc{y},\acqb{x}) {\, :^o\,} \acqb{T}$ instead of $\cc{y} {:^\tau} \cc{x}$ is easily derivable using the (=-I) rule. 
Russell's $\forall y$ is encoded by our `any' $\cc{y}$ (again, deploy ($\Pi$-I) to obtain a version of the rule in which $\Pi^\tau$, corresponding to $\forall$, is explicit).
Only Russell's $\exists$, the operator of `ontological existence', is not immediately recoverable, ($\iota$-I) thus retains the well-known oscillation between generic/maximality and particular/existential readings of descriptions, cf. e.g. \cite{ludlow2023}. 
 But if certain conditions related to $\acqb{x}$ are met, ($\Pi$-E) and ($\Sigma$-I) allow us to derive the existential reading.
($\rotatediota$-E) captures the well-known idea that the only ${F}$ is an ${F}$, 
which many writers state as an axiom but in our rule-based approach the idea is naturally presented as a rule.
($\rotatediota$-INST) differs from the other INST-rules because the function $\rotatediota$ is partial, not total, so the second premiss had to be added.
 Finally, let us stress at least one consequence of the above indicated fact that 
any application consisting of 
$\acqs{\sim} , \acqs{\supset}, \acqs{=}^\tau, \acqs{\iota}^\tau$ or $\acqs{\supset} $ and $\cc{X}$ (and $\cc{Y}$) that is $v$-improper in $\mathscr{M}$, e.g. $\acqs{=}^\tau (\cc{X}, \cc{Y}$), is  $v$-improper in $\mathscr{M}$ --
`error' is thus `propagated up',
`functions' are \textit{strict}.
An application $\acqs{\rotatediota}^\tau (\cc{C})$ is $v$-improper if $\cc{C}$ is $v$-improper: in such a case, 
$\cc{G} (\acqs{\rotatediota}^\tau (\cc{C}))$ and even 
$\acqs{=}^\iota (\acqs{\rotatediota}^\tau (\cc{C}), \acqs{\rotatediota}^\tau (\cc{C}))$ are also $v$-improper. Contra negative/positive $\mathsf{FL}$s, cf. e.g. Scott \cite{scott1979}, Feferman \cite{feferman1995}, Farmer \cite{farmer1990}, Bencivenga \cite{bencivenga2002}, Lehmann \cite{lehmann2002}, Indrzejczak \cite{indrzejczak2023}, and even Russell \cite{russell1905,whitehead-russell1910-13}.

Numerous rules are \textit{derivable} in $\mathsf{ND_{TT^*}}$.
For proof and discussion of the \textit{Rule of Existential Generalisation} (EG), see \cite{raclavsky2022, raclavsky2022-puzzleseg}; for proofs of (L-$\sim$.iii) and (L-APP.$\bot$), see \cite{kuchynka-raclavsky2024}. Let $\acqs{\bot}^{\tau_{(i)}}/\tau_{(i)}$.

\begin{center}
\small

\begin{minipage}{.3\textwidth}
\begin{prooftree}
\AxiomC{$\Gamma \longrightarrow \cc{C}_{(\cc{X}/\cc{x})} {:^o \;} \acqb{T} $}
\RightLabel{(EG)}
\UnaryInfC{$\Gamma \longrightarrow \acqs{\Sigma}^\tau (\lambda \cc{x} . (\cc{C}(\cc{x}) )) {:^o \;} \acqb{T}  $}
\end{prooftree}
\end{minipage}
\begin{minipage}{.3\textwidth}
\begin{prooftree}
\AxiomC{$\Gamma \longrightarrow 
\acqs{\sim} (\cc{o}) {\; {:}^{o} \; } \acqb{T}$}
\RightLabel{(L-$\sim$.iii)}
\UnaryInfC{$\Gamma \longrightarrow 
\cc{o} {\; {:}^{o} \; } \acqb{F}$}
\end{prooftree}
\end{minipage}
\begin{minipage}{.3\textwidth}
\begin{prooftree}
\small
\AxiomC{$\Gamma \longrightarrow 
\cc{X}_i {\; {:}^{ \tau_i} \; } \acqs{\bot}$}
\RightLabel{(L-APP.$\bot$)}
\UnaryInfC{$\Gamma \longrightarrow 
\cc{Y} (\vec{\cc{X}}) 
{\; {:}^{\tau} \; } \acqs{\bot}
$}
\end{prooftree}
\end{minipage}

\end{center}

\section{Applications to reasoning framed within natural language $\mathsf{NL}$}

The above $\mathsf{TT^*}$ can be extended to endorse various methods of \textit{natural language processing} ($\mathsf{NLP}$), e.g. Tich\'{y}'s \textit{transparent intensional logic} (\textit{TIL}) (e.g. \cite{tichy1988}), or its more effective variant \textit{transparent hyperintensional logic} (\textit{THL}) 
proposed by Kuchy\v{n}ka (p.c.) and developed in Raclavsk\'y 
\cite{raclavsky2020}. For simplicity reasons we use a simplified TIL here (with only one, alethic \textit{modality}) which is rather close to THL. 
For that sake let $\mathscr{T}$ be extended
by the atomic type $\omega$ such that $\mathscr{D}_\omega$
consists of (primitive) entities $\mathtt{w}_1, \mathtt{w}_2, ...$, called \textit{possible worlds}.
($\acqs{\Pi}^\omega$ and $\acqs{\Sigma}^\omega$ may serve as modal operators.)
The meanings of $\mathsf{NL}$ expressions are constructions of $\mathsf{TT^*}$. 
In case of expressions whose \textit{reference} varies dependently \textit{on} $\obj{w}$, the meanings in question are constructions of \textit{possible-worlds} \textit{intensions}, i.e. total or partial functions $\mathscr{D}_\omega {\to} \mathscr{D}_\tau$. 

\textit{Propositions} are intensions with $\tau := o$;
\textit{properties} (or \textit{$m$-ary relations-in-intensions}) of $\tau_1$-objects are intensions with $\tau := (\tau_1 {\to} o)$ (or $\tau := \langle \vec{\tau } \rangle {\to} o$); 
\textit{individual offices} are intensions with $\tau := \iota$; \textit{offices of individual offices} are intensions with $\tau := \omega{\to}\iota$, etc. 
 The well-known PWS-style notion of \textit{individual concepts} was adjusted by Tich\'{y} to his notion of individual offices as total/partial functions from $\langle$possible world, time instant$\rangle$ couples.
 We simplify the concept here due to the omission of time-instants parameter.\footnote{Unlike the original notion of {individual concepts}, Tich\'{y} repeatedly attempted to provide philosophical elucidations of offices, see esp. his papers ``Individuals and their Roles'' and ``Existence and God'' in \cite{tichy2004} and, of course, his \cite{tichy1988}.}

 To simplify things, (declarative) sentences are assumed to express $o$-\textit{constructions}, i.e. constructions of truth values, not propositions; they typically contain a free \textit{possible world variable} $\cc w$, i.e. $\cc{w}/\omega$. 
For further simplification, instead of constructions of properties, we will often deploy $\cc{F}$ such that $\cc {F}/ \iota{\to}o$.

With Tich\'{y} we maintain that (typical empirical, definite) \textit{descriptions} ``$\cc D$'' \textit{of individuals} (such as e.g. ``\textsl{the King of England/France}'') 
express constructions of individual offices; i.e. $\cc{D}/\omega{\to}\iota$. In many cases, $\cc D$ is a complex construction, often involving the iota operator $\acqs{\rotatediota}^\tau$; for examples, see below.
(Analogously for other types of definite descriptions.)
Note carefully that the meaning of description ``$\cc{D}$'' is the construction $\cc{D}$, not an office $\obj{D}$. On the other hand, the denotation of an \textit{empirical description} ``$\cc{D}$'' is an office $\obj{D}$ and the value of $\obj{D}$ in $\obj{w}$ is called the \textit{reference} of ``$\cc{D}$'' \textit{in} $\obj{w}$ -- while Tich\'y used an apt term \textit{occupant} of $\obj{D}$ in $\cc{w}$.
In case of \textit{non-empirical descriptions} ``$\cc{D}$'' such as e.g. ``\textsl{the only number $n$ such that $n= 3 \div 1$}'' we usually got rid of dull functional dependence on $\obj{w}$, ``$\cc{D}$'''s denotation is thus not an office, but simply its constant value, which is thus not distinguished from ``$\cc{D}$'''s reference.
Recall that each (well-formed) description always has a meaning: in case of non-empirical descriptions it is a construction of a $\tau$-object (if any), in case of empirical descriptions it is a construction of a $\tau$-office whose value in given $\obj{w}$ is a $\tau$-object (if any).

Sentences such as 
\begin{center}
``\textsl{The $\cc{D}$ is an $\cc{F}$.}'' (``\textsl{The $\cc{D}$ is in $\cc{R}$ with $\cc{D}'$}'' etc.) 
\end{center}
have often two readings, called \textit{extensional and intensional reading} (it is surprising that such useful distinction evaporated from recent philosophical logic)
\\
(i) \quad In their \textit{extensional reading},
such sentences are aptly paraphrased as
\begin{quote}
``\textsl{The occupant of the office $\cc D$ is an $\cc F$}''.
\end{quote}
In such a reading they express an $o$-construction $\cc{S}$ in which the construction $\cc D$ of the office occurs as applied to $\cc w$, i.e. $\cc{D} (\cc{w})$, which is abbreviated to
$
\cc{D}_\cc{w}
$
($\cc D_{\cc w}/ \iota$ if $\cc D/ \omega{\to}\iota$). (Similarly for other types of  expressions denoting intensions.)
If there is \textit{no} \textit{reference} of ``$\cc D$'' in $\mathtt{w}$, as in the case of ``\textsl{the King of France}'',
sentences involving them typically
\textit{lack a truth value}.
\\
(ii)  \quad 
In their \textit{intensional reading}, 
such sentences are aptly paraphrased as
\begin{quote}
``\textsl{The office $\cc D$ is an $\cc F$}''.
\end{quote}
In such a reading they express an $o$-construction $\cc{S}$ in which the 
construction $\cc{D}$ is not so applied. I.e., the subject of such an assertion is the individual office \textit{per se}, not its occupant in $\obj{w}$ (as in the extensional case). 
(Similarly for other types of  expressions denoting intensions.)
The type of reading is often indicated by the predicate; to illustrate, let ``$\cc{D}$'' be ``\textsl{the US president}'': 
if ``$\cc{F}$'' is ``\textsl{to be blue-eyed}'', i.e. a predicate applicable to individuals, not offices, one naturally renders ``\textsl{The $\cc{D}$ is $\cc{F}$}'' in the extensional sense; 
if ``$\cc{F}$'' is ``\textsl{to be one of the highest offices}'', i.e. a predicate applicable to offices, not individuals, one naturally renders ``\textsl{The $\cc{D}$ is $\cc{F}$}'' in the intensional sense.

\textit{Examples}. Recapitulation of some type annotations added or changed in this section: $ \cc{x}, \cc{y} / \iota$ ($\mathscr{D}_\iota$ consists of individuals); $\cc{D}/ \omega {\to}\iota; \cc F/ \iota{\to} o; 
\acqs{=}^\iota / \langle \iota, \iota, \rangle {\to} o; 
\acqs{\rotatediota}^\iota / (\iota{\to}o ) {\to} \iota;
 \acqb{T}, \acqb{F} , \acqs{\bot}^o / o; \cc{w} / \omega$. 

\vspace{-10pt}

\begin{center}
\begin{tabular}{ll|ll}

\textit{\small Expression for extension} & \textit{its meaning/type} &
 \textit{expression for intension} & \textit{its meaning/type} \\
\hline

``\textsl{be self-identical}'' & $\lambda \cc{x} .\acqs{=}^\iota (\cc{x}, \cc{x}) / \iota {\to} o$ &
``\textsl{be bald}'' & $\acqb{B}
/ \omega{\to} (\iota {\to} o)$ \\

``\textsl{be identical with}'' & $\acqb{=}^\iota / \langle \iota, \iota \rangle {\to} o$ &
``\textsl{be the King of sth.}'' & $\acqb{K} / \omega{\to}(\langle \iota, \iota \rangle {\to} o)$
 \\

``\textsl{France}'' & $\acqb{Fr} / \iota$ &
``\textsl{the King of France}'' 
& $\lambda \cc{w} . \acqs{\rotatediota}^\iota (\lambda \cc{x} . \acqb{K}_\cc{w} (\cc{x} , \acqb{Fr}) ) / \omega{\to}\iota$ \\

\end{tabular}
\end{center}	

\noindent
\;\
The sentence ``\textsl{The King of France is bald.}'' expresses the $o$-construction \,
$\acqb{B}_\cc{w} (\acqs{\rotatediota}^\iota (\lambda \cc{x} . \acqb{K}_\cc{w} (\cc{x} , \acqb{Fr}) )) $.

\textit{Validity} of $\mathsf{NL}$ \textit{arguments} $\mathtt{A}$ such as

\begin{prooftree}
\AxiomC{``\textsl{The King of France is identical with Louis.}''}
\UnaryInfC{``\textsl{Louis is a King of France}.''}
\end{prooftree}

\noindent
is proof-theoretically \textit{justified} 
by showing a (derived) rule $\mathtt{R}$ of $\mathsf{ND_{TT^*}}$
(where $\acqb{L}/\iota$):

\begin{prooftree}
\AxiomC{$\Gamma \longrightarrow 
\acqs{=}^\iota
(\acqs{\rotatediota}^\iota (\lambda \cc{x} . \acqb{K}_\cc{w} (\cc{x} , \acqb{Fr}) ), \acqb{L}) {\; :^o\;} \acqb{T}$}
\RightLabel{(L.=.Desc-E) (an instance of)}
\UnaryInfC{$\Gamma \longrightarrow 
\acqb{K}_\cc{w} (\acqb{L} , \acqb{Fr})  {\; :^o\;} \acqb{T}$}
\end{prooftree}

\noindent
such that (i) each formalisation (meaning) $\cc{P}_1, ..., \cc{P}_n$ 
of $\mathtt{A}$'s premisses is matched with $\acqb{T}$
(i.e. $\cc{P}_i {\; :^o \; } \acqb{T}$, for each $1 \leq i \leq m$), forming thus the succedents of $\mathtt{R}$'s premisses, while (ii) the formalization of $\mathtt{A}$'s conclusion is matched with $\acqb{T}$, too, forming thus succedent of $\mathtt{R}$'s conclusion. 
(Equivalently, the set of $\mathtt{R}$'s premisses is empty, but all  matches $\cc{P}_i {\; :^o \; } \acqb{T}$ 
occur on the left of $\longrightarrow$ as antecedents in $\mathtt{R}$'s conclusion.)

\begin{proof}[Proof of (the instance of) (L.=.Desc-E)]
$\newline$

\vspace{-20pt}
\begin{prooftree}
\AxiomC{$\Gamma \longrightarrow 
\acqs{=}^\iota
(\acqs{\rotatediota}^\iota (\lambda \cc{x} . \acqb{K}_\cc{w} (\cc{x} , \acqb{Fr}) ), \acqb{L}) {\; :^o\;} \acqb{T}$}

\RightLabel{($=$-E)}

\UnaryInfC{$\Gamma \longrightarrow 
\acqs{\rotatediota}^\iota (\lambda \cc{x} . \acqb{K}_\cc{w} (\cc{x} , \acqb{Fr}) ) {\; :^\iota \;} \acqb{L}$}

\RightLabel{($\rotatediota$-E)}

\UnaryInfC{$\Gamma \longrightarrow 
[\lambda \cc{x} . \acqb{K}_\cc{w} (\cc{x} , \acqb{Fr}) ] (\acqb{L}) {\; :^o \;} \acqb{T}$}

\RightLabel{($\beta$-CON)}

\UnaryInfC{$\Gamma \longrightarrow 
\acqb{K}_\cc{w} (\acqb{L} , \acqb{Fr})  {\; :^o\;} \acqb{T}$}
\end{prooftree}
\end{proof}

\vspace{-10pt}

\subsection{Case: Intensional Transitives}

\textit{Intensional transitive (verbs)} (ITV) 
are 
verbs such as ``\textsl{seek}'', ``\textsl{looking for}'', ``\textsl{wish [being  something]}''; they attribute a connection to \textit{agents} and  \textit{objects} of \textit{intentional attitudes}. 
In this paper, we will put aside all ITVs such as ``\textsl{believe}'', ``\textsl{know}'',  ``\textsl{wish [that]}'' whose sentential complements are sentences, forming thus sentences called \textit{reports} of \textit{propositional attitudes} (for their investigation, see e.g. our \cite{raclavsky2020}).

Since Church \cite{church1951} and Quine \cite{quine1956}, who discussed examples such as ``\textsl{Ponce de León searched for the Fountain of Youth}'', it's widely held that \textit{object terms} complementing ITVs only serve to indicate to which \textit{notion} (not material object) an agent is intentionally related to.
For not only that there's no point in e.g. looking for an object to which an agent is already consciously related to: sometimes the sought object under the description needn't to exist.

This gives rise to two widely accepted observations, (1) and (2).

\noindent
(1) \quad
{\it
Sentences with object terms in the scope of ITVs 
lack existential import as regards them.
}

\noindent
For example, the following type of arguments is obviously invalid (as indicated by {\scriptsize $----$}):\footnote{As noted by Church \cite{church1951}, Russell's
theory of descriptions blatantly fails here, since 
(unlike in the case of propositional-attitudes reports), only primary occurrence elimination of the description is possible here, so the 
unwelcome conclusion is derivable.}

\vspace{-5pt}
\begin{prooftree}
\AxiomC{\textsl{Ponce de Le\'{o}n searched for the Fountain of Youth.}}
\dashedLine
\UnaryInfC{\textsl{The Fountain of Youth exists.}}
\end{prooftree}

\noindent
(2) \quad
{\it
Substitution for object terms in the scope of ITVs fails.
}

\noindent
For an example, consider a so-called \textit{hidden description}
``\textsl{Endora}'' and:

\vspace{-5pt}
\begin{prooftree}
\AxiomC{\textsl{Ponce de Le\'{o}n seeks the Fountain of Youth.}}
\noLine
\UnaryInfC{\textsl{Endora is the Fountain of Youth.}}
\dashedLine
\UnaryInfC{\textsl{Ponce de Le\'{o}n seeks Endora.}}
\end{prooftree}

A natural choice for fulfilment of the requirements (1) -- (2)
is to employ Fregean \textit{modes of presentations} (\textit{senses}), explained in the 
Carnapian \cite{carnap1947} spirit as possible-worlds intensions called (say) \textit{individual concepts}. Explaining thus ITVs as denoting relations(-in-intension) between agents and the individual concepts.
Montague (e.g. \cite{montague1973}) is famous for this, but 
Tich\'{y}'s proposal (cf. e.g. \cite{tichy2004,tichy1988}) is more elaborated: his  offices (i) can be partial functions (such offices are \textit{unoccupied} in the respective worlds $\mathtt{w}$), (ii) they are functions from $\langle $possible world, time instant$\rangle$ couples (which we simplify in this paper), (iii) and systematically occur even in extensional contexts (via constructions $\cc{D}$ applied to $\cc{w}$, i.e. $\cc{D}_\cc{w}$).

It remains to explain why the above two arguments fail.
Let $\acqb{S} / \omega {\to} (\omega {\to} \iota)$ (searched for); 
$\acqb{FY} /\omega {\to} \iota$ (for simplicity); $\acqb{L} /\iota $ (León).
The (major) premiss of the arguments illustrating (1) and (2) expresses
\[ \cc{P} := \acqb{S}_\cc{w} (\acqb{L}, \acqb{FY}).\]
To $\cc{P}$, one cannot apply the type-theoretical version of (EG) that targets $\iota$-constructions such as $\acqb{FY}_\cc{w}$, 
since they're missing in $\cc{P}$. 
The only applicable version of (EG) (as regards object terms) targets constructions of individual offices, here $\acqb{FY}$. 
Then, one may only infer the uninformative 
$\acqs{\Sigma}^{\omega{\to}\iota} (  \acqs{=}^{\omega{\to}\iota} ( \cc{d},\acqb{FY} )) $, where $\cc{d}/ {\omega{\to}\iota}$,
expressed by 
``\textsl{There is an individual office of the Fountain of Youth}''.

Similarly for the argument illustrating point (2). Let us adjust (SI) (proved in \cite{tichy1986}) to two versions:

\begin{center}
\small
\begin{minipage}{.5\textwidth}
\begin{prooftree}
\AxiomC{$\Gamma \longrightarrow \cc{S}_{[\cc{D}_\cc{w}/\cc{x}]} {\; :^o\; } \acqb{T}$}
\AxiomC{\hspace{-10pt}$\Gamma \longrightarrow \acqs{=}^\iota (\cc{D}_\cc{w}, \cc{D}'_\cc{w}) {\; :^o\; } \acqb{T}$}
\RightLabel{(SI$_1$)}
\BinaryInfC{$\Gamma \longrightarrow \cc{S}_{[\cc{D}'_\cc{w}/\cc{x}]}
{\; :^o\; } \acqb{T} $}
\end{prooftree}
\end{minipage}
\begin{minipage}{.49\textwidth}
\begin{prooftree}
\AxiomC{$\Gamma \longrightarrow \cc{S}_{[\cc{D} /\cc{d}]} {\; :^o\; } \acqb{T}$}
\AxiomC{\hspace{-10pt}$\Gamma \longrightarrow \acqs{=}^{\omega{\to}\iota} (\cc{D}, \cc{D}') {\; :^o\; } \acqb{T}$}
\RightLabel{(SI$_2$)}
\BinaryInfC{$\Gamma \longrightarrow \cc{S}_{[\cc{D}'/\cc{d}]}
{\; :^o\; } \acqb{T} $}
\end{prooftree}
\end{minipage}
\end{center}

\noindent
The rule (SI$_1$), which uses a non-trivial \textit{co-reference} identity statement, 
cannot be applied in our case 
(for $\cc{P}$ doesn't contain $\acqb{FY}_\cc{w}$, but mere $\acqb{FY}$).
Only (SI$_2$) is applicable. But since according to  (SI$_2$)'s second premiss
``$\cc{D}$'' is \textit{co-denotative} with ``$\cc{D}'$'', one only changes the names of one and the same office $\obj{D}$ that is reportedly the object of the agent's attitude.

\subsection{Case: Strawsonian Reasoning about Existential Presuppositions}

By its design, $\mathsf{ND_{TT^*}}$ is powerful in capturing reasoning about partiality. It is then no surprise that it allows formalization of Strawson's famous views concerning 
\textit{existential presuppositions} (as indicated in \cite{raclavsky2011}). 
Recall that these are sentences ``$\cc{E}$'' ascribing existence to some object, if any, fitting the description ``$\cc{D}$''
that must be true in order
the sentences ``$\cc{S}$''  in which ``$\cc{D}$'' is in `\textit{referential position}' be either true, or false 
-- not without a truth value.
If, on the other hand, ``$\cc{E}$'' is false, the corresponding
``$\cc{S}$'' is without a true value (being \textit{gappy}).
We'll consider three arguments concerning ``$\cc{E}$''s.

\noindent
($\mathtt{A}_1$) 
\qquad
On p. 330 of Strawson's \cite{strawson1950}, we find two formulations of the following argument (let ``\textsl{the KF}'' abbreviate ``\textsl{the King of France}''):

\begin{prooftree}
\AxiomC{\textsl{The sentence ``{\rm The KF is (not) bald}'' has a truth value (true or false).}}
\UnaryInfC{\textsl{
The sentence ``{\rm The KF exists}'' is true}.}
\end{prooftree}

\noindent
Setting aside its meta-linguistic mode, assume the argument as an inference is captured by 

\begin{prooftree}
\AxiomC{$\Gamma \longrightarrow
[\acqs{\sim}] \cc{F} (\cc{D}_\cc{w}) {\; :^o\;} \cc{o}$}
\UnaryInfC{$\Gamma \longrightarrow
\acqs{\Sigma}^\iota  (\lambda \cc{x} . \acqs{=}^\iota ( \cc{D}_\cc{w}) , \cc{x}) {\; :^o\;} \acqb{T}$}
\end{prooftree}

The following derived rule of $\mathsf{ND_{TT^*}}$,
which we will call  the
\textit{Strawsonian Presupposition Rule $1$} (SPR1), covers it (recall that $\acqb{o}$ is either $\cc{o}, \acqb{T}$, or $\acqb{F}$).

\begin{theorem}
The following is a \textit{derived rule} of $\mathsf{ND}_\mathsf{TT^*}$:
\begin{prooftree}
\AxiomC{$\Gamma \longrightarrow
\cc{F} (\cc{D}_\cc{w}) {\; :^o\;} \acqb{o}$}
\RightLabel{(SPR1)}
\UnaryInfC{$\Gamma \longrightarrow
\acqs{\Sigma}^\iota  (\lambda \cc{x} . \acqs{=}^\iota ( \cc{D}_\cc{w}, \cc{x}) {\; :^o\;} \acqb{T}$}
\end{prooftree}
\end{theorem}


\begin{proof}
We begin with an assumption introduced by (AX) that fits the premiss that $\cc{F} (\cc{D}_\cc{w}) $ is $v$-proper:

\vspace{-15pt}

\begin{prooftree}

\AxiomC{}
\RightLabel{(AX)}
\UnaryInfC{$\Gamma , \cc{D}_\cc{w} {\; :^\iota\;} \cc{x}
\longrightarrow 
\cc{D}_\cc{w} {\; :^\iota\;} \cc{x}
$}

\RightLabel{(=-I)}
\UnaryInfC{$\Gamma , \cc{D}_\cc{w} {\; :^\iota\;} \cc{x}
\longrightarrow 
\acqs{=}^\iota (\cc{D}_\cc{w} , \cc{x}) 
{\; :^o\;} \acqb{T}
$}
\AxiomC{}
\RightLabel{(TM)}
\UnaryInfC{$\Gamma 
\longrightarrow \cc{x} {\; :^\iota \;} \cc{x} $}

\RightLabel{($\beta$-EXP)}
\BinaryInfC{$\Gamma , \cc{D}_\cc{w} {\; :^\iota\;} \cc{x}
\longrightarrow 
[\lambda \cc{x} . \acqs{=}^\iota (\cc{D}_\cc{w} , \cc{x})] (\cc{x}) 
{\; :^o\;} \acqb{T}
$}

\RightLabel{($\Sigma$-I)}
\UnaryInfC{
$\Gamma , \cc{D}_\cc{w} {\; :^\iota\;} \cc{x}
\longrightarrow 
\acqs{\Sigma}^\iota (\lambda \cc{x} .
\acqs{=}^\iota (\cc{D}_\cc{w} , \cc{x}) )
{\; :^o\;} \acqb{T}$}

\RightLabel{(WR)}
\UnaryInfC{$\Gamma , \cc{D}_\cc{w} {\; :^\iota\;} \cc{x} , \cc{F} {\; :^{\iota{\to}o}\;}  \cc{f}
\longrightarrow 
\acqs{\Sigma}^\iota (\lambda \cc{x} .
\acqs{=}^\iota (\cc{D}_\cc{w} , \cc{x}) )
{\; :^o\;} \acqb{T}$}

\RightLabel{(WR)}
\AxiomC{$\Gamma \longrightarrow 
\cc{F} (\cc{D}_\cc{w}) {\; :^o\;} \acqb{o}$
}

\RightLabel{(a-INST)}
\BinaryInfC{$\Gamma 
\longrightarrow 
\acqs{\Sigma}^\iota (\lambda \cc{x} .
\acqs{=}^\iota (\cc{D}_\cc{w} , \cc{x}) )
{\; :^o\;} \acqb{T}$}

\end{prooftree}

\end{proof}

\vspace{-10pt}

\noindent
($\mathtt{A}_2$)  \qquad
On p. 330 of Strawson's \cite{strawson1950}, one also finds an argument  resembling to:

\begin{prooftree}
\AxiomC{\textsl{The sentence ``{\rm The KF doesn't exist}'' is true.}}
\UnaryInfC{\textsl{The sentence ``{\rm The KF is bald}'' is without a truth value}.}
\end{prooftree}

\noindent
The argument can be seen as justified by (what we call)
\textit{Strawsonian Presupposition Rule $2$} (SPR2).\footnote{To really justify the above argument, one should derive the conclusion $\Gamma \longrightarrow  \acqs{\sim} \acqs{\Sigma}^\iota 
(\lambda \cc{o} . 
\acqs{=}^\iota ( \cc{F} (\cc{D}_\cc{w}), \cc{o}) )
{\; :^o\;} \acqb{T}$, using (SPR3) and (L-$\sim$.iii) on (SPR2)'s actual conclusion.
}

\begin{theorem}
The following is a \textit{derived rule} of $\mathsf{ND}_\mathsf{TT^*}$:
\begin{prooftree}
\AxiomC{$\Gamma \longrightarrow 
\acqs{\sim} (\acqs{\Sigma}^\iota  (\lambda \cc{x} . \acqs{=}^\iota (  \cc{D}_\cc{w}, \cc{x}) )) {\; :^o\;} \acqb{T}$}
\RightLabel{(SPR2)}
\UnaryInfC{$\Gamma \longrightarrow  \cc{F} (\cc{D}_\cc{w}) {\; :^o\;} \acqs{\bot}$}
\end{prooftree}
\end{theorem}

\noindent

\vspace{-5pt}
\begin{proof}
To simplify the proof presentation, let's first state auxiliary matches $\mathtt{M}_1, \mathtt{M}_1$ and derivation $\mathtt{D}_1$:

\vspace{-5pt}
\noindent
\begin{center}
\noindent
\begin{minipage}{.35\textwidth}
$\mathtt{M}_1 :=  \acqs{\Sigma}^\iota  (\lambda \cc{x} . \acqs{=}^\iota ( \cc{D}_\cc{w}, \cc{x} ) ) {\; :^o\;} 
\cc{o} $
\end{minipage}
\begin{minipage}{.20\textwidth}
$\mathtt{M}_2 := \cc{D}_\cc{w} {\; :^\iota \;} \cc{x}$
\end{minipage}
\begin{minipage}{.42\textwidth}
\begin{prooftree}
\small
\AxiomC{}
\LeftLabel{$\mathtt{D}_1:=$}
\RightLabel{(AX)}
\UnaryInfC{$\Gamma  , \mathtt{M}_1
\longrightarrow 
\acqs{\Sigma}^\iota  (\lambda \cc{x} . \acqs{=}^\iota ( \cc{D}_\cc{w}, \cc{x} ) ) {\; :^o\;} \cc{o}$}
\end{prooftree}
\end{minipage}
\end{center}

\vspace{-0pt}
\noindent
Derivation $\mathtt{D}$. 
Now we  develop the left branch $\mathtt{D}$ of the whole proof tree:
\vspace{-10pt}

\noindent
\begin{prooftree}
\small

\AxiomC{$\Gamma  
\longrightarrow 
\acqs{\sim} (\acqs{\Sigma}^\iota  (\lambda \cc{x} . \acqs{=}^\iota ( \cc{D}_\cc{w}, \cc{x} ) )) {\; :^o\;} \acqb{T}$}

\RightLabel{(WR)}
\UnaryInfC{$\Gamma  , \mathtt{M}_1
\longrightarrow 
\acqs{\sim} (\acqs{\Sigma}^\iota  (\lambda \cc{x} . \acqs{=}^\iota ( \cc{D}_\cc{w}, \cc{x} ) )) {\; :^o\;} \acqb{T}$}

\AxiomC{$\mathtt{D}_1$}

\RightLabel{(a-SUB)}
\BinaryInfC{$\Gamma  , \mathtt{M}_1
\longrightarrow 
\acqs{\sim} (\cc{o}) {\; :^o\;} 
\acqb{T}$
}

\RightLabel{(L-$\sim$.iii)}
\UnaryInfC{$\Gamma    , \mathtt{M}_1
\longrightarrow 
\cc{o} {\; :^o\;} 
\acqb{F}
$}

\AxiomC{}
\RightLabel{(TM)}
\UnaryInfC{$\Gamma   
\longrightarrow 
\cc{o} {\; :^o\;} 
\cc{o}
$}

\RightLabel{(WR)}
\UnaryInfC{$\Gamma  , \mathtt{M}_1
\longrightarrow 
\cc{o} {\; :^o\;} 
\cc{o}
$}

\RightLabel{($\beta$-EXP)}
\BinaryInfC{$\Gamma , \mathtt{M}_1
\longrightarrow 
[\lambda \cc{o} . \cc{o}] (\cc{o}) {\; :^o\;} 
\acqb{F}
$}

\AxiomC{$\mathtt{D}_1$}

\RightLabel{(a-SUB)}
\BinaryInfC{$\Gamma , \mathtt{M}_1
\longrightarrow 
[\lambda \cc{o} . \cc{o}] (
\acqs{\Sigma}^\iota  (\lambda \cc{x} . \acqs{=}^\iota ( \cc{D}_\cc{w}, \cc{x} ) ) 
)
{\; :^o\;} 
\acqb{F}
$}

\RightLabel{($\beta$-CON)}
\UnaryInfC{$\Gamma , \mathtt{M}_1
\longrightarrow 
\acqs{\Sigma}^\iota  (\lambda \cc{x} . \acqs{=}^\iota ( \cc{D}_\cc{w}, \cc{x} ) ) {\; :^o\;} 
\acqb{F}
$}

\RightLabel{($\Sigma$-INST)}
\UnaryInfC{$\Gamma 
\longrightarrow 
\acqs{\Sigma}^\iota  (\lambda \cc{x} . \acqs{=}^\iota ( \cc{D}_\cc{w}, \cc{x} ) ) {\; :^o\;} 
\acqb{F}
$}

\RightLabel{(WR)}
\UnaryInfC{$\Gamma , \mathtt{M}_2
\longrightarrow 
\acqs{\Sigma}^\iota  (\lambda \cc{x} . \acqs{=}^\iota ( \cc{D}_\cc{w}, \cc{x} ) ) {\; :^o\;} 
\acqb{F}
$}

\end{prooftree}

\vspace{-0pt}
\noindent
In the middle branch, 
an assumption \textit{per absurdum} that ``$\cc{D}$'' is referring in $\mathtt{w}$ is introduced by (AX):

\vspace{-15pt}

\begin{prooftree}
\small

\AxiomC{$\mathtt{D}$}

\AxiomC{}

\RightLabel{(AX)}
\UnaryInfC{$\Gamma , \mathtt{M}_2
\longrightarrow 
\cc{D}_\cc{w} {\; :^\iota \;} 
\cc{x} $}

\RightLabel{($=$-I)}
\UnaryInfC{$\Gamma , \mathtt{M}_2
\longrightarrow 
\acqs{=}^\iota ( \cc{D}_\cc{w}, \cc{x} )
{\; :^o \;} \acqb{T} $}

\AxiomC{}
\RightLabel{(TM)}
\UnaryInfC{$\Gamma \longrightarrow 
\cc{x} {\; :^\iota \;} \cc{x} $}

\RightLabel{(WR)}
\UnaryInfC{$\Gamma , \mathtt{M}_2
\longrightarrow 
\cc{x} {\; :^\iota \;} \cc{x} $}


\RightLabel{($\beta$-EXP)}
\BinaryInfC{$\Gamma , \mathtt{M}_2
\longrightarrow 
[\lambda \cc{x} . \acqs{=}^\iota ( \cc{D}_\cc{w}, \cc{x} ) ]
(\cc{x} ) {\; :^o\;} \acqb{T} $}

\RightLabel{($\Sigma$-I)}
\UnaryInfC{$\Gamma , \mathtt{M}_2
\longrightarrow 
\acqs{\Sigma}^\iota  (\lambda \cc{x} . \acqs{=}^\iota ( \cc{D}_\cc{w}, \cc{x} ) ) {\; :^o\;} 
\acqb{T}
$}


\RightLabel{(EFQ)}
\BinaryInfC{$\Gamma , \mathtt{M}_2
\longrightarrow 
\cc{D}_\cc{w} {\; :^\iota \;} 
\acqs{\bot} $}

\AxiomC{}
\RightLabel{(AX)}
\UnaryInfC{$\Gamma  , \cc{D}_\cc{w} {\; :^\iota \;} 
\acqs{\bot} 
\longrightarrow 
\cc{D}_\cc{w} {\; :^\iota \;} 
\acqs{\bot}$}

\RightLabel{(EXH)}
\BinaryInfC{$\Gamma 
\longrightarrow 
\cc{D}_\cc{w} {\; :^\iota \;} 
\acqs{\bot} $}

\RightLabel{(L-APP$^\bot$.ii)}
\UnaryInfC{$\Gamma \longrightarrow 
\cc{F} (\cc{D}_\cc{w})  
{\; :^o \;} 
\acqs{\bot} $}

\end{prooftree}
\end{proof}

\vspace{-10pt}

\noindent
($\mathtt{A}_3$) \qquad
On p. 331 of Strawson's \cite{strawson1950},
we find an argument quite fitting the rule (L-APP$^\bot$.ii). Let us rather study a justification of an argument which looks
like an inverse of $\mathtt{A}_2$.

\begin{prooftree}
\AxiomC{\textsl{The sentence ``{\rm The KF is (not) bald}'' is without a truth value.}}
\UnaryInfC{\textsl{The sentence ``{\rm The KF exists}'' is false}.}
\end{prooftree}

\noindent
It can be seen as justified by (what we call) the \textit{Strawsonian Presupposition Rule $3$} (SPR3).

\begin{theorem}
The following is a \textit{derived rule} of $\mathsf{ND}_\mathsf{TT^*}$:

\begin{prooftree}
\AxiomC{$\Gamma \longrightarrow 
\cc{F} (\cc{D}_\cc{w}) {\; :^o\;} \acqs{\bot}$}
\AxiomC{$\Gamma \longrightarrow 
\cc{F} (\cc{y}) {\; :^o\;} \cc{o}$}
\RightLabel{(SPR3)}
\BinaryInfC{$\Gamma \longrightarrow 
\acqs{\Sigma}^\iota  (\lambda \cc{x} . \acqs{=}^\iota ( \cc{x} , \cc{D}_\cc{w}) ) {\; :^o\;} \acqb{F}$}
\end{prooftree}

\end{theorem}

\noindent
\textit{Remark.} In (SPR3)'s second premiss, we require $\cc F$ $v$-constructs  a \textit{total} characteristic function (in $\mathscr{M}$). For in cases when $\cc F$ $v$-constructed a partial characteristic function (in $\mathscr{M}$), the whole application $\cc{F} (\cc{D}_\cc{w}) $ would also be $v$-improper (in $\mathscr{M}$), so we couldn't derive (SPR)'s conclusion for sure.\footnote{To justify the above argument, the first premiss of the rule should be converted to $\Gamma \longrightarrow 
\acqs{\sim} \acqs{\Sigma}^\iota  (\lambda \cc{x} . \acqs{=}^\iota ( \cc{x} , \cc{D}_\cc{w}) ) {\; :^o\;} \acqb{T}$.}

(SPR3)'s proof (occurring in the end of this section) becomes simple, once two derived rules are established.

\begin{lemma}
The following is a \textit{derived rule} of $\mathsf{ND_{TT^*}}$:

\begin{prooftree}
\AxiomC{$\Gamma \longrightarrow 
\cc{F} (\cc{D}_\cc{w}) {\; :^o\;} \acqs{\bot}$}
\AxiomC{$\Gamma \longrightarrow 
\cc{F} (\cc{y}) {\; :^o\;} \cc{o}$}
\RightLabel{(L-Desc$^\bot$-APP)}
\BinaryInfC{$\Gamma \longrightarrow  \cc{D}_\cc{w} {:^\iota \;} \acqs{\bot} $}
\end{prooftree}

\end{lemma}

\begin{proof} 
An assumption \textit{per absurdum} is introduced by (AX) 
in the right middle branch. First, auxiliary derivations $\mathtt{D}_1$ and $\mathtt{D}_2$ are stated:

\vspace{3pt}

\begin{center}
\begin{minipage}{.4\textwidth}
\vspace{-10pt}
\begin{prooftree}
\AxiomC{$\Gamma \longrightarrow 
\cc{F} (\cc{D}_\cc{w}) {\; :^o\;} \acqs{\bot}$}
\LeftLabel{$\mathtt{D}_1 :=$}
\RightLabel{(WR)}
\UnaryInfC{$\Gamma , \cc{D}_\cc{w} {:^\iota \;} \cc{y} \longrightarrow 
\cc{F} (\cc{D}_\cc{w}) {\; :^o\;} \acqs{\bot}$}
\end{prooftree}
\end{minipage}
\begin{minipage}{.4\textwidth}
\begin{prooftree}
\AxiomC{}
\LeftLabel{$\mathtt{D}_2 :=$}
\RightLabel{(AX)}
\UnaryInfC{$\Gamma , \cc{D}_\cc{w} {:^\iota \;} \acqs{\bot} 
\longrightarrow 
\cc{D}_\cc{w} {:^\iota \;} \acqs{\bot} $}
\end{prooftree}
\end{minipage}
\end{center}

\vspace{-10pt}

\begin{center}

\begin{prooftree}
\AxiomC{$\mathtt{D}_1$}


\AxiomC{$\Gamma \longrightarrow 
\cc{F} (\cc{y}) {\; :^o\;} \cc{o}$}
\RightLabel{(WR)}
\UnaryInfC{$\Gamma , \cc{D}_\cc{w} {:^\iota \;} \cc{y} 
\longrightarrow 
\cc{F} (\cc{y}) {\; :^o\;} \cc{o}$}

\AxiomC{}
\RightLabel{(AX)}
\UnaryInfC{$\Gamma , \cc{D}_\cc{w} {:^\iota \;} \cc{y}
\longrightarrow 
\cc{D}_\cc{w} {:^\iota \;} \cc{y}$}

\RightLabel{(a-SUB)}
\BinaryInfC{$\Gamma , \cc{D}_\cc{w} {:^\iota \;} \cc{y}
\longrightarrow 
\cc{F} (\cc{D}_\cc{w}) {:^\iota \;} \cc{o}$}

\RightLabel{(EFQ)}
\BinaryInfC{$\Gamma , \cc{D}_\cc{w} {:^\iota \;} \cc{y} 
\longrightarrow \cc{D}_\cc{w} {:^\iota \;} \acqs{\bot} $}

\AxiomC{$\mathtt{D}_2$}

\RightLabel{(EXH)}
\BinaryInfC{$\Gamma \longrightarrow  \cc{D}_\cc{w} {:^\iota \;} \acqs{\bot} $}

\end{prooftree}

\end{center}
\end{proof}

\vspace{-10pt}
\begin{lemma}
The following is a \textit{derived rule} of $\mathsf{ND_{TT^*}}$:

\vspace{-5pt}
\begin{prooftree}
\AxiomC{$\Gamma \longrightarrow
\cc{D}_\cc{w} {\; :^\iota\;} \acqs{\bot}
$}
\RightLabel{(L-$\Sigma$.Desc$^\bot$-APP)}
\UnaryInfC{$\Gamma \longrightarrow 
\acqs{\Sigma}^\iota  (\lambda \cc{x} . \acqs{=}^\iota ( \cc{x} , \cc{D}_\cc{w}) ) {\; :^o\;} \acqb{F}$}
\end{prooftree}
\end{lemma}

\begin{proof}
The presentation of the proof is split in three pieces. First, auxiliary matches are stated:

\vspace{3pt}

\begin{center}
\begin{minipage}{.35\textwidth}
$\mathtt{M}_1 :=
\acqs{\Sigma}^\tau (\lambda \cc{x} . \acqs{=}^\iota (\cc{D}_\cc{w}, \cc{x})) {\; :^o\;} \cc{o}$
\end{minipage}
\begin{minipage}{.2\textwidth}
$\mathtt{M}_2 := \cc{o} {\; :^o\;} \acqb{T} $
\end{minipage}
\begin{minipage}{.4\textwidth}
$\mathtt{M}_3 :=
[\lambda \cc{x}  . \acqs{=}^\iota (\cc{D}_\cc{w}, \cc{x})] (\cc{x}) {\; :^o\;} \acqb{T}$
\end{minipage}
\end{center}

\noindent
Derivation $\mathtt{D}_1$. We begin with the assumption \textit{per absurdum} that it is true that an individual $\mathtt{y}$ belongs to the (one-membered) set of individuals who are the reference of ``$\cc{D}$'' in $\mathtt{w}$ (cf. $\mathtt{M}_3$). This will suggest that the truth value of the 
relevant existence ascription (cf. $\mathtt{M}_1$) is $\mathtt{F}$.

\begin{center}
\begin{prooftree}

\AxiomC{$\Gamma \longrightarrow
\cc{D}_\cc{w} {\; :^\iota\;} \acqs{\bot}
$}
\RightLabel{(WR)}
\UnaryInfC{$\Gamma , \mathtt{M}_3 \longrightarrow
\cc{D}_\cc{w} {\; :^\iota\;} \acqs{\bot} $}

\AxiomC{}
\RightLabel{(AX)}
\UnaryInfC{$\Gamma , \mathtt{M}_3
\longrightarrow 
[\lambda \cc{x}  . \acqs{=}^\iota (\cc{D}_\cc{w}, \cc{x})] (\cc{x}) {\; :^o\;} \acqb{T}$}
\RightLabel{($\beta$-CON)}
\UnaryInfC{$\Gamma , \mathtt{M}_3
\longrightarrow 
\acqs{=}^\iota (\cc{D}_\cc{w}, \cc{x}) {\; :^o\;} \acqb{T}$}
\RightLabel{($=$-E)}
\UnaryInfC{$\Gamma , \mathtt{M}_3
\longrightarrow 
\cc{D}_\cc{w} {\; :^\iota\;} \cc{x}$}

\RightLabel{(EFQ)}
\BinaryInfC{$\Gamma , \mathtt{M}_3 \longrightarrow 
\cc{o} {\; :^o\;} \acqb{F}$}

\RightLabel{(WR)}
\UnaryInfC{$\Gamma , \mathtt{M}_1 , \mathtt{M}_2 , \mathtt{M}_3 \longrightarrow 
\cc{o} {\; :^o\;} \acqb{F}$}

\end{prooftree}
\end{center}

\noindent
Derivation $\mathtt{D}_2$. 
Now we elaborate the redundant assumption (below, we'll therefore use (RA)) that the truth value of the relevant existence ascription is $\mathtt{T}$ (cf. $\mathtt{M}_2$ and the left middle branch).

\vspace{-10pt}
\begin{center}
\begin{prooftree}
\small

\AxiomC{$\mathtt{D}_1$}

\AxiomC{}
\RightLabel{(AX)}
\UnaryInfC{$\Gamma , \mathtt{M}_2 \longrightarrow
\cc{o} {\; :^o\;} \acqb{T}$}

\AxiomC{}
\RightLabel{(TM)}
\UnaryInfC{$\Gamma  \longrightarrow 
\cc{o} {\; :^o\;} \cc{o} $}
\RightLabel{(WR)}
\UnaryInfC{$\Gamma , \mathtt{M}_2 \longrightarrow 
\cc{o} {\; :^o\;} \cc{o}  $}

\RightLabel{($\beta$-EXP)\!\!}
\BinaryInfC{$\Gamma , \mathtt{M}_2 \longrightarrow 
[\lambda \cc{o} . \cc{o}] (\cc{o}) {\; :^o\;} \acqb{T} $}

\RightLabel{(WR)\!\!}
\UnaryInfC{$\Gamma , \mathtt{M}_1 , \mathtt{M}_2 \longrightarrow
[\lambda \cc{o} . \cc{o}] (\cc{o}) {\; :^o\;} \acqb{T} $}

\AxiomC{}
\RightLabel{(AX)}
\UnaryInfC{$\Gamma , \mathtt{M}_1 \longrightarrow
\acqs{\Sigma}^\iota (\lambda \cc{x} . \acqs{=}^\iota (\cc{D}_\cc{w}, \cc{x})) {\; :^o\;} \cc{o}$}
\RightLabel{(WR)}
\UnaryInfC{$\Gamma , \mathtt{M}_1 , \mathtt{M}_2 \longrightarrow
\acqs{\Sigma}^\iota (\lambda \cc{x} . \acqs{=}^\iota (\cc{D}_\cc{w}, \cc{x})) {\; :^o\;} \cc{o}$}

\RightLabel{(a-SUB)}
\BinaryInfC{$\Gamma , \mathtt{M}_1 , \mathtt{M}_2 \longrightarrow
[\lambda \cc{o} . \cc{o} ] (\acqs{\Sigma}^\iota (\lambda \cc{x} . \acqs{=}^\iota (\cc{D}_\cc{w}, \cc{x})))
 {\; :^o\;} \acqb{T}$}

\RightLabel{($\beta$-CON)}
\UnaryInfC{$\Gamma , \mathtt{M}_1 , \mathtt{M}_2 \longrightarrow
\acqs{\Sigma}^\iota (\lambda \cc{x} . \acqs{=}^\iota (\cc{D}_\cc{w}, \cc{x})) {\; :^o\;} \acqb{T}$}


\RightLabel{($\Sigma$-E)}
\BinaryInfC{$\Gamma ,\mathtt{M}_1 , \mathtt{M}_2 \longrightarrow 
\cc{o} {\; :^o\;} \acqb{F}$}

\end{prooftree}
\end{center}

\noindent
Finally, we put the truth value $\mathtt{F}$ with the existence ascription together (first, we eliminate $\mathtt{M}_2$, cf. left).

\vspace{-10pt}
\begin{center}
\begin{prooftree}
\small

\AxiomC{$\mathtt{D}_2$\!\!\!\!\!\!}

\AxiomC{}
\RightLabel{(AX)}
\UnaryInfC{$\Gamma , \cc{o} {\; :^o\;} \acqb{F}
\longrightarrow 
\cc{o} {\; :^o\;} \acqb{F}
$}
\RightLabel{(WR)}
\UnaryInfC{$\Gamma , \mathtt{M}_1 , 
\cc{o} {\; :^o\;} \acqb{F}
\longrightarrow 
\cc{o} {\; :^o\;} \acqb{F}
$}

\RightLabel{(RA)}
\BinaryInfC{$\Gamma , \mathtt{M}_1 
\longrightarrow 
\cc{o} {\; :^o\;} \acqb{F}
$}

\AxiomC{}
\RightLabel{(TM)}
\UnaryInfC{$\Gamma \longrightarrow
\cc{o} {\; :^o\;} \cc{o}$}
\RightLabel{(WR)}
\UnaryInfC{$\Gamma , \mathtt{M}_1 \longrightarrow
\cc{o} {\; :^o\;} \cc{o}$}

\RightLabel{($\beta$-EXP)\!\!\!\!}
\BinaryInfC{$\Gamma  , \mathtt{M}_1 \longrightarrow 
[\lambda \cc{o}. \cc{o}]
(\cc{o}) {\; :^o\;} \acqb{F}$}

\AxiomC{}
\RightLabel{(AX)}
\UnaryInfC{$\Gamma , \mathtt{M}_1 \longrightarrow \acqs{\Sigma}^\iota (\lambda \cc{x} . \acqs{=}^\iota (\cc{D}_\cc{w}, \cc{x})) {\; :^o\;} \cc{o}$}

\RightLabel{(a-SUB)}
\BinaryInfC{$\Gamma , \mathtt{M}_1 \longrightarrow 
[\lambda \cc{o}. \cc{o}] (
\acqs{\Sigma}^\iota  (\lambda \cc{x} . \acqs{=}^\iota ( \cc{x} , \cc{D}_\cc{w}) ) )  {\; :^o\;} \acqb{F}$}

\RightLabel{($\beta$-CON)}
\UnaryInfC{$\Gamma , \mathtt{M}_1 \longrightarrow 
\acqs{\Sigma}^\iota  (\lambda \cc{x} . \acqs{=}^\iota ( \cc{x} , \cc{D}_\cc{w}) ) {\; :^o\;} \acqb{F}$}

\RightLabel{($\Sigma$-INST)}
\UnaryInfC{$\Gamma \longrightarrow 
\acqs{\Sigma}^\iota  (\lambda \cc{x} . \acqs{=}^\iota ( \cc{x} , \cc{D}_\cc{w}) ) {\; :^o\;} \acqb{F}$}

\end{prooftree}
\end{center}

\end{proof}

\vspace{-10pt}
\begin{proof}[Proof of (SPR3)]
$\newline$
\vspace{-30pt}

\begin{prooftree}
\AxiomC{$\Gamma \longrightarrow 
\cc{F} (\cc{D}_\cc{w}) {\; :^o\;} \acqs{\bot}$}

\AxiomC{$\Gamma \longrightarrow 
\cc{F} (\cc{x}) {\; :^o\;} \cc{o}$}

\RightLabel{(L-Desc$^\bot$-APP)}
\BinaryInfC{$\Gamma \longrightarrow
\cc{D}_\cc{w} {\; :^\iota\;} \acqs{\bot}
$}

\RightLabel{(L-$\Sigma$.Desc$^\bot$-APP)}
\UnaryInfC{$\Gamma \longrightarrow 
\acqs{\Sigma}^\iota  (\lambda \cc{x} . \acqs{=}^\iota ( \cc{x} , \cc{D}_\cc{w}) ) {\; :^o\;} \acqb{F}$}
\end{prooftree}

\end{proof}

\vspace{-15pt}
\section{Conclusion}

We exposed a specific theory of definite descriptions in Tich\'{y}an spirit whose essential features were listed in Sec. 1. The derivation rules of the partial type theory $\mathsf{TT^*}$ that govern the $\rotatediota$-operator were exposed and briefly discussed in Sec. 3. In Sec. 4, we showed its application in natural language processing, in particular
to two famous cases of reasoning: 
(a) 
the case with intensional transitives whose complements are non-referring descriptions, and 
(b) 
the case of Strawsonian rules for existential presuppositions concerning non-referring descriptions -- which have not been studied in a formal way in literature.
Future work should focus more on 
(i) proof-theoretic properties of the above $\rotatediota$-rules and (ii) comparison with rival logical approaches both in free and modal logic (cf. \cite{indrzejczak2023,indrzejczak-zawidzki2023,orlandelli2021}).

\medskip

\noindent
\textit{Acknowledgment.}
The present author thanks 
to reviewers for many helpful suggestions and 
to Petr Kuchyňka for useful remarks and essentially his proof of (SPR2).


\nocite{*}

\bibliographystyle{eptcs}
\bibliography{generic}

\end{document}